%% file: tmp-arXiv.tex
\documentclass{article}
\usepackage[utf8]{inputenc}

\usepackage{graphicx}	
\usepackage{epstopdf}
\usepackage[english,vlined,ruled]{algorithm2e}	
\usepackage{amsthm}			
\usepackage{amssymb}		
\usepackage{amsmath}		
\newtheorem{theorem}{Theorem}[section]
\newtheorem{lemma}[theorem]{Lemma}
\newtheorem{proposition}[theorem]{Proposition}
\newtheorem{fact}{Fact}
\newtheorem{claim}{Claim}
\newtheorem{property}{Property}

\newcommand{\comment}[1]{}

\title{Approximation of the Double Travelling Salesman Problem with Multiple Stacks}
	\author{Laurent Alfandari, Sophie Toulouse}
\date{July 2020}

\begin{document}

\maketitle

\begin{abstract}
	The Double Travelling Salesman Problem with Multiple Stacks, $\mathsf{DTSPMS}$, deals with the collect and delivery of $n$ commodities in two distinct cities, where the pickup and the delivery tours are related by LIFO constraints. During the pickup tour, commodities are loaded into a container of $k$ rows, or stacks, with capacity $c$. This paper focuses on computational aspects of the $\mathsf{DTSPMS}$, which is NP-hard. 

	We first review the complexity of two critical subproblems: deciding whether a given pair of pickup and delivery tours is feasible and, given a loading plan, finding an optimal pair of pickup and delivery tours, are both polynomial under some conditions on $k$ and $c$. 
	
	We then prove a $(3k)/2$ standard approximation for the $\mathsf{Min\,Metric\,k\,DTSPMS}$, where $k$ is a universal constant, and other approximation results for various versions of the problem. 

	We finally present a matching-based heuristic for the $\mathsf{2\,DTSPMS}$, which is a special case with $k=2$ rows, when the distances are symmetric. 
	This yields a $1/2-o(1)$, $3/4-o(1)$ and $3/2+o(1)$ standard approximation for respectively $\mathsf{Max\,2\,DTSPMS}$, its restriction $\mathsf{Max\,2\,DTSPMS-(1,2)}$ with distances $1$ and $2$, and $\mathsf{Min\,2\,DTSPMS-(1,2)}$, and a $1/2-o(1)$ differential approximation for $\mathsf{Min\,2\,DTSPMS}$ and $\mathsf{Max\,2\,DTSPMS}$.
\end{abstract}

\section{Introduction and problem statement}\label{sec-intro}

The {\em Double Travelling Salesman Problem with Multiple Stacks}, $\mathsf{DTSPMS}$, was introduced in \cite{PM09}. It was initially motivated 
by the request of a transportation company that owns a fleet of vehicles and a depot in $m$ distinct cities. For any pair $(C,C')$ of cities, the company handles orders from the suppliers of city $C$ to the customers of city $C'$, proceeding as follows. In city $C$, a single vehicle picks up all the orders to be delivered to a customer in $C'$, storing them into a single container. Then the container is sent to $C'$, where a local vehicle delivers the commodities. The operator is thus faced with two levels of transportation: the local routing inside the cities, and the long-haul transportation between the cities. 

The $\mathsf{DTSPMS}$ addresses the local routing problem: given a set of orders from suppliers of $C$ to customers of $C'$, it aims at finding an optimal pair of tours $(T_P,T_D)$, where $T_P$ is a pickup tour on $C$ and $T_D$ is a delivery tour on $C'$. The value of a solution $(T_P,T_D)$ is the sum of the costs of tours $T_P$ and $T_D$. The problem specificity relies on the way the containers are managed: a container consists of a given number of rows that can be accessed only by their front side, and no reloading plan is allowed. The rows of the container are thus subject to {\em Last In First Out} (LIFO) rules that constrain the delivery tour. Namely, for a given row, the latter has no other choice but to handle the goods stored in this row in the precise reverse order than the one performed by the pickup tour.

\smallskip
Note that the {\em Double Travelling Salesman Problem with Multiple Stacks} is also known as the {\em Multiple Stacks Travelling Salesman Problem} ($\mathsf{STSP}$) in the litterature, \cite{TW09,T10,T12}. 


\bigskip
The $\mathsf{DTSPMS}$ has strong connections with the {\em Traveling Salesman Problem} ($\mathsf{TSP}$).
Given a node set $V$ and a distance function $d:V^2\rightarrow\mathbb{Q}^+$, the $\mathsf{TSP}$ consists in finding a tour $T$ that visits every node $i\in V$ exactly once, and whose total distance $d(T)=\sum_{e\in T}d(e)$ is minimized, or maximized in some cases. 

Formally, an instance $I=(n,k,c,d_P,d_D,\mathrm{opt})$ of the $\mathsf{DTSPMS}$ consists of a number $n$ of orders, a number $k$ of rows of the container, a capacity $c\geq \lceil n/k\rceil$ which is the same for all rows of the container, two distance functions $d_P,d_D: V^2\rightarrow\mathbb{Q}^+$ where $V =\{0,1\,\ldots,n\}$, and an optimization goal $\mathrm{opt}\in\{\min,\max\}$. Indices $i=1,\ldots,n$ refer to the orders (and thus to the location of the associated supplier in city $C$ and the associated customer in city $C'$), whereas index $0$ refers to the depot (and thus to the location of the depot in cities $C$ and $C'$). Functions $d_P,d_D$ satisfy $d_P(i, i) =d_D(i, i) =0, i\in V$.
 %
We consider instances $I_P=(V,d_P)$, $I_D=(V,d_D)$ and $I_\Sigma=(V,d_\Sigma)$ of the $\mathsf{TSP}$, where $d_\Sigma$ is defined as 
$$\begin{array}{rll}
	d_\Sigma(i,j)	&= d_P(i,j)+d_D(j,i),		&i,j\in V
\end{array}$$
and the goal on $I_P,I_D,I_\Sigma$ is to minimize {\em iff} the goal is to minimize on $I$.

On the one hand, a pair $(T_P,T_D)$ of pickup and delivery tours is feasible for the $\mathsf{DTSPMS}$ on $I$ {\em iff} there exists a loading plan of the commodities into the rows of the container that satisfies the following property: 
\begin{property}\label{pty-feasible}
For every pair of orders $(i,j)$ such that $T_P$ handles $i$ before $j$ in the pickup tour, then either $T_D$ handles $j$ before $i$, or commodities $i$ and $j$ are loaded in two distinct rows.
\end{property}
Note that such a loading plan always exists when $k=n$ and in that case the problem consists of solving two independent $\mathsf{TSP}$. 
On the other hand, given a tour $T=(0,i_1,\ldots,i_n,0)$, let $T^-$ denotes its reverse tour, {\em i.e.}, $T^-=(0,i_n,\ldots,i_1,0)$.  
Then the pair $(T,T^-)$ of pickup and delivery tours is a feasible solution for the $\mathsf{DTSPMS}$ on $I$ for all $k,c$. 
For example, one feasible loading plan consists of putting commodities $i_1,\ldots,i_c$ in the first row, then commodities $i_{c+1},\ldots,i_{2c}$ in the second row, and so on.

Hence, solving the $\mathsf{TSP}$ independently on $I_D$ and $I_P$ can be seen as a relaxation of the $\mathsf{DTSPMS}$ on $I$, and solving $\mathsf{DTSPMS}$ on $I$ provides a relaxation of the $\mathsf{TSP}$ on $I_\Sigma$. 
In particular, when $k=n$, the $\mathsf{DTSPMS}$ on $I$ reduces to solving independently the $\mathsf{TSP}$ on $I_D$ and $I_P$, and when $k=1$, the $\mathsf{DTSPMS}$ on $I$ reduces to the $\mathsf{TSP}$ on $I_\Sigma$. Therefore the  $\mathsf{DTSPMS}$ is $\mathbf{NP-hard}$. 

\smallskip
Notice that $\mathsf{DTSPMS}$ is different from the {\em Pickup and Delivery Travelling Salesman Problem} ($\mathsf{PD\,TSP}$), where pickups and deliveries are operated during the same tour. Indeed, the $\mathsf{PD\,TSP}$ is a {\em Travelling Salesman Problem with Precedence Constraints} ( $\mathsf{PC\,TSP}$) where the set of precedence constraints consists of a perfect matching on $V\backslash\{0\}$, while in the $\mathsf{DTSPMS}$, the computation of an optimal pickup tour (or of an optimal delivery tour) when the loading plan is given is a $\mathsf{PC\,TSP}$ where the precedence constraints partition $V=\{1,\ldots,n\}$ into $k$ strict orders (see section \ref{sec-cpx-P}).

\bigskip
Before closing this introduction by the organization of the paper, we now define some variants of the $\mathsf{DTSPMS}$. When dealing with routing problems, various assumptions can be made on the distance functions. They may be symmetric, or satisfy the triangular inequalities (metric case), or take values in $\{a,b\}$ for two reals $a\neq b$. 
A distance $d$ on a vertex set $V$ is symmetric when $d(u, v) =d(v, u), u, v\in V$.
It is metric provided that $d(u, v) \leq d(u, w) +d(w, v), u, v, w\in V$.
The corresponding restrictions of the $\mathsf{DTSPMS}$ are denoted by Symmetric $\mathsf{DTSPMS}$, $\mathsf{Metric\,DTSPMS}$ and $\mathsf{DTSPMS-(a,b)}$, respectively. 

We denote by $\mathsf{k\,DTSPMS}$ the restriction of the $\mathsf{DTSPMS}$ where the number $k$ of rows is a universal constant. We call the instances when $c\geq n$ as uncapacitated, and tight when $c=\lceil n/k\rceil$.

\bigskip
Many algorithms have already been proposed for the $\mathsf{DTSPMS}$: see \cite{PM09,FOT09,felipe2009double,urrutia2015dynamic} for metaheuristic approaches, \cite{LLER10} (where the authors solve the $\mathsf{2\,DTSPMS}$ by means of the $k$ best pickup tours and the $k$ best delivery tours), \cite{PAS10,AMCDI13,alba2013branch, barbato2016polyhedral} (branch and cut algorithms for the $\mathsf{DTSPMS}$) or \cite{CCS13} (a branch and bound algorithm for the $\mathsf{2\,DTSPMS}$) for exact approaches. Several methods have also been proposed for some generalizations of the problem such as the Pick-up and delivery TSP with multiple stacks \cite{sampaio2017new} and VRP (multiple tours) with multiple stacks \cite{iori2015exact, chagas2020variable}. In this paper we focus on the computational aspects of the $\mathsf{DTSPMS}$.
The paper is organized as follows: 
\begin{itemize}
	
	\item Section \ref{sec-cpx} reviews complexity aspects of two subproblems. Solutions of the $\mathsf{DTSPMS}$ consist of a pair of pickup and delivery tours on $V$ and a loading plan on $\{1,\ldots,n\}$. If one fixes the former, then deciding whether it admits a feasible loading plan or not consists of an instance of the {\em bounded coloring problem}, $\mathsf{BC}$ in permutation graphs (section \ref{sec-cpx-tour}). If one fixes the latter, then optimizing the pickup or delivery tour with respect to the considered loading plan consists of an instance of the {\em $\mathsf{TSP}$ with precedence constraints}, $\mathsf{PC\,TSP}$ where the precedence constraints partition $\{1,\ldots,n\}$ into at most $k$ linear orders (section \ref{sec-cpx-P}). These two subproblems of the $\mathsf{DTSPMS}$ may turn to be tractable, depending on $k$ and $c$.
	\item Section \ref{sec-tsp} provides approximation results with the standard ratio for the $\mathsf{DTSPMS}$ with a general number $k$ of rows, exploring connections to the $\mathsf{TSP}$. 
	We first compare the optimal value of the $\mathsf{DTSPMS}$ to optimal values of the $\mathsf{TSP}$ with distances $d_P$, $d_D$ and $d_P+d_D$, in the general and metric cases. We then derive both positive and negative approximability results for the $\mathsf{DTSPMS}$, in particular we show that Hamiltonian cycles that are $3/2$--standard approximate for the $\mathsf{Min\,Metric\,TSP}$ and the distance $d_P+d_D$ are $3$--standard approximate for the $\mathsf{Min\,Metric\,DTSPMS}$.
	\item Section \ref{sec-apx} provides standard approximation ratios for the symmetric $\mathsf{2\,DTSPMS}$, {\em i.e.} $k=2$ and distances are symmetric. We present a matching-based heuristic which yields a standard approximation ratio of $1/2-o(1)$, $3/4-o(1)$ and $3/2+o(1)$ for respectively $\mathsf{Max\,2\,DTSPMS}$, $\mathsf{Max\,2\,DTSPMS-(1,2)}$ and $\mathsf{Min\,2\,DTSPMS-(1,2)}$.
	
	\item Section \ref{sec-apx-dapx} is strictly devoted to differential approximation. The main result is a differential approximation ratio of $1/2-o(1)$ for the $\mathsf{2\,DTSPMS}$, obtained by adapting the matching heuristic of the previous section.
\end{itemize}
Section \ref{sec-conc} finally concludes the paper with some perspectives. 


\section{Complexity of $\mathsf{DTSPMS}$ subproblems} \label{sec-cpx}


A solution of the $\mathsf{DTSPMS}$ consists of two parts: the loading plan, and the pair of pickup and delivery tours. We here investigate the subproblems obtained when fixing either the loading plan, or the pair of pickup and delivery tours.

\subsection{Deciding feasibility of a pair of tours}\label{sec-cpx-tour}

A pair of tours $T_P$ and $T_D$ is a feasible pair of pickup and delivery tours {\em iff} every pair of nodes $(i,j)$ that are visited in the same order in $T_P$ and $T_D$ can be loaded in two distinct rows (Property \ref{pty-feasible}). 
Equivalently, let $G'=(V\backslash\{0\},E')$ be the graph where $E'$ consists of the pairs of commodities $i,j$ such that $T_P$ and $T_D$ visit $i,j$ in the same order; then $(T_P,T_D)$ is a feasible pair of pickup and delivery tours {\em iff} $V\backslash\{0\}$ can be partitioned into $k$ independent sets of size at most $c$ in $G'$. 
Indeed, assume that there exists such a partition $W_1,\ldots,W_k$ of $V\backslash\{0\}$.
We note the pick-up tour $T_P=(0,i_1,\ldots,i_n,0)$. One can build a feasible loading plan the following way: 
iteratively for $p=1,\ldots,n$, load the node $i_p$ at the end of the row $r$ such that $i_p \in W_r$.  
By definition of $E'$, such a loading plan admits $T_D$ as a delivery tour compatible with $T_P$: when two nodes $i\neq j\in V\backslash\{0\}$ are loaded in a same row $r$, either $T_P$ visits $i$ before $j$ and $T_D$ visits $j$ before $i$, or $T_P$ visits $j$ before $i$ and $T_D$ visits $i$ before $j$. 
Hence, every non-empty row $r\in\{1,\ldots,k\}$ consists of a subsequence $(i'_1,\ldots,i'_\ell)$ of $(i_1,\ldots,i_n)$ that is visited in the reverse order $(i'_l,\ldots,i'_\ell)$ in $T_D$.

Now $G'$ is a permutation graph (consider the ordering either $i_1,\ldots,i_n$ or $j_1,\ldots,j_n$ of $V\backslash\{0\}$) and thus, finding a coloring of $G'$ using at most $k$ colors, if such a coloring exists, requires a $O(n\log n)$ time (see \cite{PLE71}). 
Consequently, deciding the feasibility of a pair of tours for the uncapacited $\mathsf{DTSPMS}$ is in $\mathbf{P}$. 
Furthermore, the {\em Bounded Coloring problem}, $\mathsf{BC}$ consists, given a graph $G$ and two integers $k,c$ such that $kc\geq n$, in assigning to every node $i\in V(G)$ a color $f(v)\in \{1,\ldots,k\}$ in such a way that:
\begin{enumerate}
	\item $f$ is a coloring, that is, two nodes $i,j\in V(G)$ such that $(i,j)\in E(G)$ are assigned a different color;
	\item every color $r\in \{1,\ldots,k\}$ is assigned to at most $c$ vertices of $V(G)$.
\end{enumerate} 
The $\mathsf{BC}$ problem is also known as the {\em Mutual Exclusion Scheduling} problem in the literature. This problem was proved to be $\mathbf{NP-complete}$ in permutation graphs given any universal constant $c\geq 6$, \cite{J03}.
Bonomo et al. extended the $\mathsf{BC}$ problem to the {\em ``Capacitated Coloring problem''}, $\mathsf{CC}$ where a maximum size $c_r$ is specified for each color $r$. They proposed a $O(n^{k^2+k+1}k^3)$-time algorithm to solve the $\mathsf{CC}$ problem in co-comparability graphs. Consequently, deciding the feasibility of a pair of tours for the $\mathsf{k\,DTSPMS}$ is in $\mathbf{P}$.
Proposition \ref{prop-coloration} summarizes the former discussion.
\begin{proposition}[\cite{CCN12,TW09,BMO11}]\label{prop-coloration}
Given a pair $(T_P,T_D)$ of tours, deciding whether it admits or not a consistent loading plan and designing such a loading plan if the answer is YES is:
\begin{enumerate}
	\item in $\mathbf{P}$ when $c=n$ (regardless of $k$) \cite{CCN12,TW09} or $k$ is a universal constant \cite{BMO11};
	\item $\mathbf{NP-complete}$ when $c$ is a universal constant greater than $5$  \cite{BMO11}.
\end{enumerate}
\end{proposition}


\subsection{Optimizing the tours for a given loading plan}\label{sec-cpx-P}

Let $\mathcal{P}$ be a loading plan of $\{1,\ldots,n\}$. 
A tour $T$ on $V$ is a consistent pickup tour with respect to $\mathcal{P}$ {\em iff} for any pair of nodes $(i,j)$ such that $\mathcal{P}$ loads $i$ at a lower position than $j$ in a same row $r$, $T$ visits $i$ before $j$. Symmetrically, $T$ on $V$ is a delivery tour consistent with $\mathcal{P}$ {\em iff} for any two distinct nodes $i,j$ such that, in some row $r$, $\mathcal{P}$ loads $i$ before $j$, $T$ visits $j$ before $i$. Hence, the problem of computing an optimal tour with respect to $\mathrm{opt}$ and $d_P$ ({\em resp.}, $d_D$), among the pickup and delivery tours that are consistent with $\mathcal{P}$, is an instance of the {\em Travelling Salesman Problem with Precedence Constraints}, $\mathsf{PC\,TSP}$, where the precedence constraints partition $\{1,\ldots,n\}$ into $k$ strict orders. 

The $\mathsf{PC\,TSP}$ consists, given a binary relation $B$ on $\{1,\ldots,n\}$ encoding precedence constraints between pairs of nodes, in finding an optimal tour among those tours $T$ that satisfy: 
\begin{align}\label{eq-PCTSP}
	T \textrm{ visits }i\textrm{ before }j,	&&(i,j)\in B
\end{align}
This problem also is know as {\em ``the Minimum Setup Scheduling''} problem in the literature. Colbourn and Pulleyblank proposed in \cite{CP85} a dynamic programming procedure which runs within polynomial time when the precedence constraints define a partial order of bounded width, {\em i.e.}, when the maximum number of pairwise incomparable nodes is bounded above by a universal constant. 

Here, the set of pairs of nodes subject to precedence constraints for the pickup ({\em resp.}, delivery) tour instance induced by a loading plan $\mathcal{P}$, is denoted by $B_P$ ({\em resp.}, $B_D$):
\begin{align}
	B_P	&=\left\{(i,j)\,|\,\textrm{$\mathcal{P}$ loads $i$ in the same row as $j$ at a lower position than $j$}\right\}\nonumber\\
	B_D	&=\left\{(j,i)\,|\,(i,j)\in B_P\right\}\nonumber
\end{align}
 %
$B_P$ and $B_D$ encode two partial orders of same width at most $k$. Consequently,
\begin{proposition}[\cite{CCN12,TW09}]\label{prop-TSP}
	The computation of a best compatible pair $(T_P,T_D)$ of pickup and delivery tours with respect to a given loading plan can be done within a $O(k^2|V|^k)$ time. 
	This problem therefore is in $\mathbf{P}$ for all constant positive integer $k$ when considering the $\mathsf{k\,DTSPMS}$.
\end{proposition}

Let us comment a bit further this fact. The dynamic procedure proposed in \cite{CP85} can be seen as some application of the dynamic programming procedure for the $\mathsf{TSP}$ \cite{HK62}. Let $T^*$ refer to an optimal tour on $V$. Moreover, given a subset $W\subseteq V\backslash\{0\}$ together with a node $i\in W$, let $P^*_{W,i}$ refer to an optimal Hamiltonian path from $0$ to $i$ on $W\cup\{0\}$.
Then the optimal tour expresses as
\begin{align}\label{eq-prog_dyn-opt}
	T^*			&=\arg\mathrm{opt}_{i\in V\backslash\{0\}}\left\{d(P^*_{V\backslash\{0\},i}) +d(i,0)\right\}
\end{align}
for a collection $\{P^*_{W,i}\,|\,W\subseteq V\backslash\{0\},i\in W\}$ of paths that satisfies the system:
\begin{align}
	P^*_{\{i\},i}	&=(0,i),		&i\in V\backslash\{0\}\label{eq-prog_dyn-init}\\
	P^*_{W,i}		&=\arg\mathrm{opt}_{j\in W\backslash\{i\}}\left\{d(P^*_{W\backslash\{i\},j}) +d(j,i)\right\},
								&i\in W\subset V\backslash\{0\}\label{eq-prog_dyn-rule}
\end{align}
Any algorithm that implements this principle will generate $\sum_{p=2}^n p\binom{n}{p}=\Theta(n 2^n)$ paths $P^*_{W,i}$, while the computation of a given path $P^*_{W,i}$ requires the comparizon of $|W|-1$ values. 

Now, when considering a set $B$ of precedence constraints, not any pair $(W\subseteq V\backslash\{0\},i\in W)$ is relevant, as not any elementary path with starting vertex $0$ can be completed into a tour on $V$ that is consistent with $B$. 
Let $\pi(W)$ denote the set of vertices $i$ such that $(i,j)\in B$ for some vertex $j\in W$; then, on the one hand, $W$ must satisfy $\pi(W)\subseteq W$. On the other hand, $i$ must satisfy $i\in W\backslash\pi(W)$. 
\comment{
	On the one hand, if $\pi(W)$ denotes the set of vertices $i$ such that $(i,j)\in B$ for some vertex $j\in W$, then $W$ must satisfy $\pi(W)\subseteq W$: otherwise, not any Hamiltonian path on $W\cup\{0\}$ with starting node $0$ can be completed into a Hamiltonian path on $V$ that is consistent with $B$. 
On the other hand, any hamiltonian path on $W\cup\{0\}$ with starting node $0$ must end at some node $i$ such that $i\in W\backslash\pi(W)$. 
}%
Equivalently, we shall restrict to the pairs $(W\cup\pi(W),i)$ where $W$ is an antichain of the considered partial order and $i\in W$. 
If $k$ refers to the width of the partial order, then the number of such pairs is bounded above by $\sum_{p=1}^kp\binom{n}{p}\leq kn^k$. 
In particular when $B=B_D$ ({\em resp.}, $B=B_P$), the antichains to consider correspond to the choice of at most one element per row and the dynamic procedure runs within $O(k^2 n^k)$-time. 


\section{Standard approximation of the $\mathsf{DTSPMS}$ using reductions to the $\mathsf{TSP}$}\label{sec-tsp}

We study the relative complexity of $\mathsf{DTSPMS}$ in regards to $\mathsf{TSP}$. The question is to know in what extent $\mathsf{DTSPMS}$ is harder to solve  than the $\mathsf{TSP}$, and how far the optimal value of $\mathsf{DTSPMS}$ is from the optimal value of related $\mathsf{TSP}$. 


\subsection{Comparison of optimal values}
In what follows, we describe a loading plan as a collection $\mathcal{P}=(P_1,\ldots,P_k)$ of node-disjoint paths on $V\backslash\{0\}=\{1,\ldots,n\}$. 
Such a collection actually corresponds to a loading plan {\em iff} $(i)$ it partitions $V\backslash\{0\}$, {\em i.e.}, $\cup_{r=1}^k V(P_r)=V\backslash\{0\}$ and $(ii)$ each path $P_r$ connects at most $c$ nodes. 
The loading plan loads item $i$ in position $p$ of row $r$ {\em iff} node $i$ is in position $p$ in path $P_r$.

\smallskip
Let $I=(n,k,c,d_P,d_D,\mathrm{opt})$ be an instance of the $\mathsf{DTSPMS}$. We consider the three instances $I_P,I_D,I_\Sigma$ of the $\mathsf{TSP}$ that were introduced in section \ref{sec-intro-apx_pty}.
The following relations are trivially satisfied: 
\begin{align}\label{eq-TSP-STSP-opt}
	OPT(I_P)+OPT(I_D)\ \succeq OPT(I)\ \succeq OPT(I_\Sigma)
\end{align}

\smallskip
Now assume that distance functions $d_P$ and $d_D$ are metric, and the goal on $I$ is to minimize. Since both $d_D$ and $d_P$ satisfy the triangular inequalities, so does $d_\Sigma$, and thus $I_\Sigma$ is metric. 
The metric case provides a more accurate comparison of the optimal values on $I$ and $I_\Sigma$ than in the general case:

\begin{figure}[t]
\begin{center}
	\includegraphics{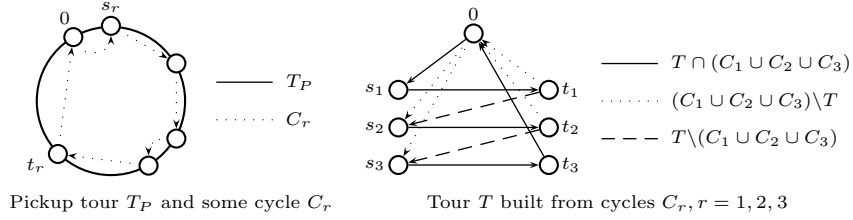}
	\caption{\label{fig-min-metric-T} The tour $T$ of Lemma \ref{lem-metric}: illustration when $h=3$}
\end{center}
\end{figure}

\begin{lemma}\label{lem-metric} For all instances $I =(n,k,c,d_P,d_D)$ of the $\mathsf{Min\,Metric\,DTSPMS}$, $I$ and its related instance $I_\Sigma$ of the $\mathsf{Min\,Metric\,TSP}$ satisfy
	\begin{align}\label{eq-metric-opt}
		OPT(I_\Sigma)	\leq k\,OPT(I) 
	\end{align}
Furthermore, relation $(\ref{eq-metric-opt})$ is asymptotically tight.
\end{lemma}

\begin{proof}
Let $(\mathcal{P},T_P,T_D)$ be any feasible solution on $I$, and let $h$ refer to the number of rows $\mathcal{P}$ actually uses. We establish inequality (\ref{eq-metric-opt}) by deriving from $(\mathcal{P},T_P,T_D)$ a solution $T$ of the $\mathsf{TSP}$ on $I_\Sigma$ with value:
\begin{align}
	d_\Sigma(T)	&\leq h\left(d_P(T_P)+d_D(T_D)\right)\label{eq-metric-reduc}
\end{align}

We assume {\em w.l.o.g} that in $\mathcal{P}$, rows $1$ to $h$ are the non-empty ones. Each row $r \in \{1,\ldots,h\}$ consists of a path $P_r$ from some node $s_r$ to some node $t_r$. We define the following cycles: 
\begin{align}\nonumber
	C_r		&=\{(0,s_r)\}\cup P_r \cup\{(t_r,0)\}, 	&r\in\{1,\ldots,h\}
\end{align} 
For any row $r\in \{1,\ldots,h\}$, $T_P$ and $P_r$ visit the nodes of $V(P_r)$ in the same order, whereas $T_D$ visits these nodes in the reverse order 
(see Figure \ref{fig-min-metric-T}). We deduce from this observation and the fact that $d_P,d_D$ satisfy the triangular inequality:
\begin{align}
	d_P(C_r) 	&\leq d_P(T_P),	&r\in \{1,\ldots,h\} \label{eq-metric-TP-C}\\
	d_D(C_r^-)	&\leq d_D(T_D),	&r \in \{1,\ldots,h\} \label{eq-metric-TD-C}
\end{align}
where $C^-$ denotes cycle $C$ in the reverse order. 
Summing relations $(\ref{eq-metric-TP-C})$ and $(\ref{eq-metric-TD-C})$ over $r=1,\ldots,h$, we obtain:
\begin{align}\label{eq-metric-TP-TD-C}
	\sum_{r=1}^{h} d_\Sigma(C_r)		&\leq h\left(d_P(T_P)+d_D(T_D)\right) 
\end{align}

Now consider the tour $T$ defined as (see Figure \ref{fig-min-metric-T}):
\begin{align}\nonumber
	T	&=\left(\bigcup_{r=1}^{h} C_r \backslash \bigcup_{r=1}^{h-1} \left\{(t_r,0),(0,s_{r+1})\right\}\right)\bigcup_{r=1}^{h-1}\left\{(t_r,s_{r+1})\right\}
\end{align}
By construction, we have:
\begin{align}\label{eq-def-T}
d_\Sigma(T) &=	
     \sum_{r=1}^h d_\Sigma(C_r)
	-\sum_{r=1}^{h-1} \left( d_\Sigma(t_r,0) +d_\Sigma(0,s_{r+1}) -d_\Sigma(t_r,s_{r+1})\right)
\end{align}
Since $d_\Sigma$ is metric, for all $r\in\{1 ,\ldots, h -1\}$, quantity $d_\Sigma(t_r,0) +d_\Sigma(0,s_{r+1}) -d_\Sigma(t_r,s_{r+1})$ is non-negative. It thus follows from relations (\ref{eq-def-T}) and  (\ref{eq-metric-TP-TD-C}) that the proposed tour $T$ indeed satisfied inequality (\ref{eq-metric-reduc}).
Relation (\ref{eq-metric-opt}) then is a straightforward consequence of inequalities 
    $OPT(I_\Sigma) \leq d_\Sigma(T)$,
    $d_\Sigma(T) \leq 	h (d_P(T_P)+d_D(T_D))$ and $h \leq k$.

\begin{figure}[t]
\begin{center}
	\includegraphics{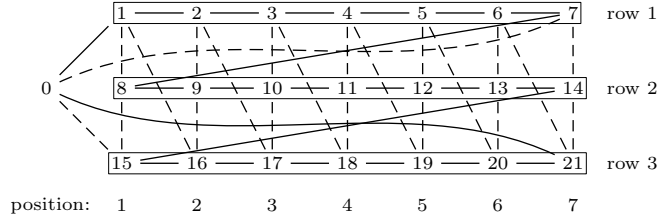}
	\caption{\label{fig-min-metric-tight} Tighness of relation (\ref{eq-metric-opt}): illustration when $k=3$ and $c=7$. Plain lines indicate edges of distance $\lambda$ for $d_P$ whereas dashed lines indicate edges of distance $1$ for $d_D$.}
\end{center}
\end{figure}
It remains us to establish that this relation is asymptotically tight.
To do so, we associate with any real number $\lambda>0$ and any two integers $k, c$ such that $c\geq k\geq 2$ an instance $I(\lambda,k,c)$ of the $\mathsf{Min\,Metric\,DTSPMS}$ on $kc +1$ vertices (including the depot vertex) with $k$ rows, each with capacity $c$. 
Let $T_P$ and $T_D$ be the tours defined by (indices are taken modulo $kc$):
$$\begin{array}{rlrl}
    T_P &=\{(i,i+1)\,|\,i\in\{0,\ldots,kc\}\}   
    &T_D &=\{(i,i+c)\,|\,i\in\{0,\ldots,kc\}\}
\end{array}$$
Then on $I(\lambda,k,c)$, 
    $d_P$ takes value $\lambda$ over $T_P$,
    $d_D$ takes value 1 $T_D$,
    $d_P$ and $d_D$ are defined by metric closure anywhere else.

By construction, $T_P$ and $T_D$ are optimal on respectively $I_P$ and $I_D$. Moreover, the pair $(T_P,T_D)$ is feasible for the $\mathsf{DTSPMS}$ on $I(\lambda,k,c)$, considering the loading plan $\mathcal{P}=(P_1,\ldots,P_k)$ defined as (see Figure \ref{fig-min-metric-tight} for some illustration):
$$\begin{array}{rll}
    P_r &=\left((r-1)c+1,(r-1)c+2,\ldots,rc\right),    &r\in\{1,\ldots,k\}
\end{array}$$
With respect to $\mathcal{P}$, $T_P$ follows paths $P_1,P_2,\ldots,P_k$ in that order, whereas $T_D$ visits the nodes according to their position in the container, from position $c$ to position $1$. 
The pair $(T_P,T_D)$ therefore defines an optimal solution on $I(\lambda,k,c)$, with value:
\begin{align}\label{eq-min-metric-tight-I}
	OPT(I(\lambda,k,c))	&=(kc+1)\times (\lambda+1)
\end{align}

Now consider instance $I_\Sigma(\lambda,k,c)$ of the $\mathsf{Min\,Metric\,TSP}$. Let $i, j$ be two distinct nodes from $\{0 ,\ldots, kc\}$. By definition of $T_P$ and $T_D$, if $i$ and $j$ are at distance lesser than $c$ in $T_P$, then they are at distance at least $k$ in $T_D$ and thus, $d_\Sigma(i, j) \geq 1 \times \lambda +k \times 1$. Otherwise, we have $d_\Sigma(i, j) \geq c \times \lambda +1 \times 1$. Furthermore, if $i$ and $j$ are adjacent in $T_P$ ({\em resp.}, in $T_D$), then they are at distance exactly $k$ ({\em resp.}, $c$) in $T_D$ ({\em resp.}, in $T_P$). We deduce that either $T_P$ or $T_D$ is optimal on $I_\Sigma(\lambda,k,c)$, depending on $c \lambda +1$ {\em versus} $\lambda +k$. The optimal value on $I_\Sigma(\lambda,k,c)$ therefore satisfies 
\begin{align}\label{eq-min-metric-tight-ISigma}
	OPT(I_\Sigma(\lambda,k,c))	
	    &= (kc+1) \times \min\{\lambda +k,\lambda c+1\}
\end{align}

In particular if $\lambda +k \leq \lambda c+1$ {\em iff} $\lambda \geq (k -1)/(c -1)$, then we have:
$$\frac{OPT(I_\Sigma(\lambda,k,c))}{OPT(I(\lambda,k,c))}
	=\frac{\lambda +k}{\lambda+1}$$
This ratio tends to $k$ when $\lambda =(k -1)/(c -1)$ and $c$ tends to $+\infty$, what concludes the proof.
\end{proof}

\subsection{Approximation results}


Approximation theory aims at providing approximate solutions of good quality for optimization problems that are hard to solve. Although we recall some definitions, we assume that the reader is familiar with the main concepts of approximation theory. If not, we refer to, {\em eg.}, \cite{ACGKMP99,DP96} for standard and differential approximation, respectively. In what follows, since one manipulates both maximization and minimization goals, we use notations $\succeq$, $\succ$, $\mathop{opt}$, $\mathop{\overline{opt}}$ instead of $\geq$, $>$, $\max$, $\min$ ({\em resp.}, $\leq$, $<$, $\min$, $\max$) if the goal is to maximize ({\em resp.}, to minimize).

Let $\Pi\in\mathbf{NPO}$ be an optimization problem and $\mathcal{I}_\Pi$ its set of instances. Given an instance $I\in \mathcal{I}_\Pi$, the value of an optimal solution on $I$ is denoted by $OPT(I)$. Finally, let $\mathcal{A}$ be an algorithm that provides feasible solutions for $\Pi$; then $APX(I)$ refers to the value of the solution output by $\mathcal{A}$ on $I\in \mathcal{I}_\Pi$. 

The standard approximation ratio compares the value $APX(I)$ of the approximate solution to the optimal value $OPT(I)$. If $\Pi$ is a maximation problem, then $\mathcal{A}$ is said to be $\rho$--approximate for some function $\rho:\mathcal{I}_\Pi\rightarrow ]0,1]$ {\em iff}
\begin{align}\nonumber
	APX(I)  &\geq \rho(I) \times OPT(I),	&I\in \mathcal{I}_\Pi
\end{align}
If the goal on $\Pi$ is to minimize, $\mathcal{A}$ is said to be $\rho$--approximate for some function $\rho:\mathcal{I}_\Pi\rightarrow [1,+\infty[$ {\em iff}
\begin{align}\nonumber
	APX(I)  &\leq \rho(I) \times OPT(I),	&I\in \mathcal{I}_\Pi
\end{align}
 %
$\Pi$ is said to be approximable within factor $\rho$ {\em iff} it admits a polynomial time $\rho$--approximation algorithm. 

For instance, within the standard approximation framework, the Symmetric $\mathsf{Max\,TSP}$ is approximable within factor $61/81-o(1)$, \cite{COW05}.
By contrast, the Symmetric $\mathsf{Min\,TSP}$ is not approximable within any constant factor. However, when restricting to metric instances, the Symmetric $\mathsf{Min\,Metric\,TSP}$ is approximable within a factor of $3/2$ (by means of the famous Christofides algorithm \cite{C76}).

Preliminary observe that the $\mathsf{TSP}$ naturally reduces to the $\mathsf{DTSPMS}$. Namely, let $I$ be an instance of the $\mathsf{TSP}$ on vertex set $V$, and pick any $v\in V$. Given any two non-negative integers $k, c$ such that $kc \geq |V| -1$, we can associate with $I$ an instance $I'$ of $\mathsf{DTSPMS}$ with $k$ rows of capacity $c$, on which vertex $v$ represents the depot vertex, and a tour $T$ on $V$ takes distances $d_P(T) =d(T)/2$ and $d_D(T) =d(T^-)/2$. Given any tour $T$ on $V$, $(T,T^-)$ is a feasible pair of pickup and delivery tours on $I'$, with value $d(T)$. Conversely, any feasible pair $(T_P,T_D)$ of pickup and delivery tours on $I'$ brings two feasible solutions $T_P$ and $T_D^-$ of $I$ whose value satisfy:
$$\begin{array}{rll}
	\min\left(d(T_P), d(T_D^-)\right)	&\leq \left(d_P(T_P) +d_D(T_D)\right)/2    
	                                    &\leq \max\left(d(T_P),d(T_D^-)\right)
\end{array}$$
Instance $I'$ of the $\mathsf{DTSPMS}$ can be seen as some generalization of instance $I$ of the $\mathsf{TSP}$ in that $I$ admits more solutions and possibly more solution values than $I$ does. However, one can with each such solution $(T_P, T_D)$ associate a solution $T$ of $I$ -- and thus, a solution $(T, T^-)$ of $I'$ -- with value at least $d_P(T_P) +d_D(T_D)$ if the goal is to maximize, at most $d_P(T_P) +d_D(T_D)$ if the goal is to minimize. We conclude that the two instances $I$ and $I'$ are equivalent to approximate. 

Thereby, known inapproximability bounds for the $\mathsf{TSP}$ also hold for the $\mathsf{DTSPMS}$. In particular, the symmetric $\mathsf{Min\,TSP}$ is not approximable within ratio $2^{-p(|V|)}$ for all polynomials $p$ unless $\mathbf{P}=\mathbf{NP}$ (folklore, but see {\em e.g.}, \cite{WS11}). The Symmetric $\mathsf{Min\,DTSPMS}$ therefore is $\mathbf{NP}-hard$ to approximate within a ratio of $O\left(2^{|V|}\right)$. 
 %
We similarly deduce from \cite{BJWW98} that $\mathsf{Max\,Metric\,DTSPMS}$ is $\mathbf{Max\,SNP}$-hard. 

$\mathsf{DTSPMS}$ reduces to $\mathsf{TSP}$ by means of a polynomial-time reduction that preserves standard approximation up to some factor for the maximization case, as well as for the bivalued case, as described in the proposition below:
\begin{proposition}\label{prop-reduc-stsp2tsp}
	The $\mathsf{DTSPMS}$ reduces to the $\mathsf{TSP}$ by means of a polynomial time reduction that maps $\rho$--standard approximate solutions of the $\mathsf{TSP}$ onto solutions of the $\mathsf{DTSPMS}$ with a standard approximation guarantee of
	\begin{enumerate} 
		\item[(i)] \label{it-apx_from_tsp-max} $\rho/2$ for $\mathsf{Max\,DTSPMS}$,
		\item[(ii)] \label{it-apx_from_tsp-max_ab} $(\rho +a/b)/2$ for $\mathsf{Max\,DTSPMS-(a,b)}$,
		\item[(iii)] \label{it-apx_from_tsp-mim_ab} $(\rho +b/a)/2$ for $\mathsf{Min\,DTSPMS-(a,b)}$.
	\end{enumerate}
	The reduction preserves the distance properties of the input instance. It thus notably maps
	symmetric, metric, 
	    $(a,b)$-valued instances of the $\mathsf{DTSPMS}$ to respectively 
    symmetric, metric, 
        $(a,b)$-valuated instances of the $\mathsf{TSP}$.
\end{proposition}

\begin{proof}
Let $I=(n,k,c,d_P,d_D,\mathrm{opt})$ be an instance of $\mathsf{DTSPMS}$. Compute a tour $T_\alpha$ on $I_\alpha$, $\alpha\in\{P,D\}$. Pick the pair 
$(T_P,T_P^-)$ of pickup and delivery tours if $d_P(T_P)+d_D(T_P^-)\succeq d_P(T_D^-)+d_D(T_D)$ and the pair 
$(T_D^-,T_D)$ otherwise. 
 
Assume that $T_\alpha$ is $\rho$--approximate for the $\mathsf{TSP}$ on $I_\alpha$, $\alpha\in\{P,D\}$. By construction, the value $APX(I)$ of the approximate solution on $I$ satisfies:
$$\begin{array}{rl}
	APX(I)		&\succeq 	1/2 \times 		\left(d_P(T_P)+d_D(T_P^-)+d_P(T_D^-)+d_D(T_D)\right)\\
				&\succeq 	\rho/2 \times	\left(OPT(I_P)+OPT(I_D)\right)	+1/2 \times 		\left(d_D(T_P^-)+d_P(T_D^-)\right)
\\
				&\succeq 	\rho/2 \times	OPT(I) +1/2 \times 		\left(d_D(T_P^-)+d_P(T_D^-)\right)
\end{array}$$ 

The result is straightforward for (i), considering  $d_D(T_P^-)+d_P(T_D^-)\geq 0$. 
As for $(ii)$ and $(iii)$, simply observe that both quantities $d_D(T_P^-) +d_P(T_D^-)$ and $OPT(I)$ express as the sum of $2(n +1)$ edge distances, that all belong to $\{a, b\}$ where $0\leq a < b$. In the maximization case, we deduce that we have:
$$\begin{array}{rlll}
    d_D(T_P^-)+d_P(T_D^-)   &\geq 2(n +1) a   
                            &\geq a/b \times 2(n +1) b 
                            &\geq a/b \times OPT(I) 
\end{array}$$
The argument for the minimization case is symmetrical.
\end{proof}
 
As for the metric case, Lemma \ref{lem-metric} yields the following conditional approximation result:
\begin{proposition}\label{prop-reduc-stsp2tsp_MinMetric}
For all constant positive integer $k$, $\mathsf{Min\,Metric\,k\,DTSPMS}$ reduces to the $\mathsf{Min\,Metric\,TSP}$ by means of a polynomial time reduction that preserves the standard approximation ratio up to a multiplicative factor of $k$.
The reduction maps symmetric instances of the $\mathsf{DTSPMS}$ to symmetric instances of the $\mathsf{TSP}$.
\end{proposition}

\begin{proof}
	Given an instance $I$ of the $\mathsf{Min\,Metric\,DTSPMS}$, first associate with $I$ instance $I_\Sigma$ of $\mathsf{Min\,Metric\,TSP}$. Then associate with any tour $T$ on $I_\Sigma$ the pair $(T,T^-)$ of pickup and delivery tours on $I$. 
	If $T$ is $\rho$--approximate on $I_\Sigma$, then the value $APX(I)$ of the approximate solution on $I$ satisfies 
    $$\begin{array}{rlll}
			APX(I)	
				&= d_\Sigma(T)
				&\leq \rho \times OPT(I_\Sigma) 
				&\leq k\rho \times OPT(I)
		\end{array}$$
	where the right-hand side inequality follows from (\ref{eq-metric-opt}).
\end{proof}

We derive from Propositions \ref{prop-reduc-stsp2tsp} and \ref{prop-reduc-stsp2tsp_MinMetric} 
the following positive 
approximation results for the $\mathsf{DTSPMS}$:

\begin{theorem}\label{thm-apx-from-tsp}
	The following bounds hold for the standard approximability of $\mathsf{DTSPMS}$:\\
	\begin{tabular}{lll}
		\multicolumn{1}{c}{Restriction}	&\multicolumn{1}{c}{Ratio}		
			&\multicolumn{1}{c}{Reference}\\\hline
		$\mathsf{Max\,DTSPMS}$						&$3/8$	    
		    &(i) of Prop. \ref{prop-reduc-stsp2tsp} \& \cite{P14}\\
		Symmetric $\mathsf{Max\,DTSPMS}$			&$7/18$
		    &(i) of Prop. \ref{prop-reduc-stsp2tsp} \& \cite{PMM09}\\[3pt]
		$\mathsf{Max\,Metric\,DTSPMS}$				&$35/88$
		    &(i) of Prop. \ref{prop-reduc-stsp2tsp} \& \cite{KM11}\\
		Symmetric $\mathsf{Max\,Metric\,DTSPMS}$	&$7/16$
		    &(i) of Prop. \ref{prop-reduc-stsp2tsp} \& \cite{KM09}\\[3pt]
		Symmetric $\mathsf{Max\,DTSPMS-(0,1)}$		&$3/7$
		    &(i) of Prop. \ref{prop-reduc-stsp2tsp} \& \cite{BK06}\\[3pt]
		$\mathsf{Min\,Metric\,DTSPMS}$				&$(2k\log |V|)/3$
		    &Prop. \ref{prop-reduc-stsp2tsp_MinMetric} \& \cite{FS07}\\
		Symmetric $\mathsf{Min\,Metric\,DTSPMS}$	&$(3k)/2$
		    &Prop. \ref{prop-reduc-stsp2tsp_MinMetric} \& \cite{C76}\\[3pt]
		$\mathsf{Min\,DTSPMS-(1,2)}$				&$13/8$
		    &(iii) of Prop. \ref{prop-reduc-stsp2tsp} \& \cite{B04}\\
		Symmetric $\mathsf{Min\,DTSPMS-(1,2)}$		&$11/7$
		    &(iii) of Prop. \ref{prop-reduc-stsp2tsp} \& \cite{BK06, adamaszek2018new}
	\end{tabular}
\end{theorem}



\section{Standard approximation of the symmetric $\mathsf{2\,DTSPMS}$}\label{sec-apx}


Thereafter, we assume that the distance functions are symmetric, the container has two rows, each of which can receive (at least) $\lceil n/2\rceil$ commodities. 

When facing routing problems, it is rather natural to manipulate optimal matchings, as they somehow bring {\em ``one half''} of the optimum value. 
We already know from Proposition \ref{thm-apx-from-tsp} that the Minimum metric case of $\mathsf{2\,DTSPMS}$ is approximable within standard factor $3$. We here present a matching-based heuristic that provides standard approximation ratios of 
$1/2 -o(1)$ for the $\mathsf{Max\,2\,DTSPMS}$, 
$1/2 \times(1+a/b) -o(1)$ for the $\mathsf{Max\,2\,DTSPMS-(a,b)}$, and 
$1/2 \times(1+b/a) +o(1)$ for the $\mathsf{Min\,2\,DTSPMS-(a,b)}$. 

\begin{algorithm}\label{algo-apx}
\caption{APX\_2DTSPMS}
\KwIn{A vertex set $V =\{0, 1, \ldots, n\}$ where $n >0$, two symmetric distance functions $d_P,d_D: V^2 \rightarrow \mathbb{Q}^+$, an optimization goal $\mathrm{opt}$}
\KwOut{A balanced $2$-rows loading plan $\mathcal{P}$ of $V\backslash\{0\}$ and an optimal pair $\left(T_P^*(\mathcal{P}),T_D^*(\mathcal{P})\right)$ of pickup and delivery tours on $V$ with respect to $d_P,d_D,\mathrm{opt}$ and $\mathcal{P}$}

\BlankLine
\ForEach{$\alpha$ in $\{P,D\}$}
{
	Compute a (near-) perfect matching $M_\alpha$ on $V$ that is optimal with respect to $d_\alpha$ and $\mathrm{opt}$\;
}

\BlankLine
$\mathcal{P}\longleftarrow \left(\emptyset,\emptyset\right)$\;
\ForEach{Connected component $W$ of the multi-edge set $M_P\cup M_D$}
{
	\If{$W$ induces a cycle $(0,v_1,\ldots,v_{2m-1},0)$}
	{
		Insert $(v_1,\ldots,v_m)$ at the beginning of row $1$ in $\mathcal{P}$\;
		Insert $(v_{2m-1},\ldots,v_{m+1})$ at the beginning of row $2$ in $\mathcal{P}$\;
	}
	\ElseIf{$W$ induces a chain $(0,v_1,\ldots,v_{2m},0)\backslash\{(v_j,v_{j+1})\}$}	
	{
		Insert $(v_1,\ldots,v_m)$ at the beginning of row $1$ in $\mathcal{P}$\;
		Insert $(v_{2m},\ldots,v_{m+1})$ at the beginning of row $2$ in $\mathcal{P}$\;
	}
	\ElseIf{$W$ induces a cycle $(v_1,\ldots,v_{2m},v_1)$}
	{
		Add $(v_1,\ldots,v_m)$ at the end of row $1$ in $\mathcal{P}$\;
		Add $(v_{2m},\ldots,v_{m+1})$ at the end of row $2$ in $\mathcal{P}$\;
	}
	\ElseIf{$W$ induces a chain $(v_1,\ldots,v_{2m+1})$}	
	{
		Add $(v_1,\ldots,v_m)$ at the end of row $1$ in $\mathcal{P}$\;
		Add $(v_{2m+1},\ldots,v_{m+1})$ at the end of row $2$ in $\mathcal{P}$\;
	}
}

\BlankLine
\ForEach{$\alpha$ in $\{P,D\}$}
{
	Compute a tour $T_\alpha^*(\mathcal{P})$ on $V$ that is consistent with $\mathcal{P}$ and optimal with respect to $d_\alpha$ and $\mathrm{opt}$\;
}

\BlankLine
\Return $(\mathcal{P},T_P^*(\mathcal{P}),T_D^*(\mathcal{P}))$;
\end{algorithm}

\subsection{The matching-based algorithm}\label{sec-apx-algo_base}

Let $I$ be an instance of the $\mathsf{2\,DTSPMS}$ and let $V$ denote the vertex set that is considered in $I$. Algorithm \ref{algo-apx} runs in three steps. 
First, it computes for $\alpha\in\{P,D\}$ a (near-) perfect matching $M_\alpha$ on $V$ that is optimal with respect to $d_\alpha$ and opt, which is well known to require a low ($\mathcal{O}({|V|}^3)$) polynomial time.

\begin{figure}[t]
\begin{center}
	\includegraphics{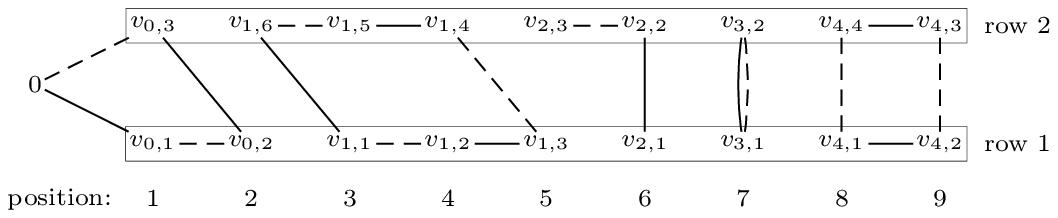}
	\caption{\label{fig-apx-P}\small{The approximate loading plan given the two (near-)perfect matchings $M_P$ (in plain lines) and $M_D$ (in dashed lines)} on $V$.}
\end{center}
\end{figure}

Second, it builds a loading plan $\mathcal{P}=(P_1,P_2)$, considering the connected components of the multigraph $(V,M_P\cup M_D)$ one after each other, starting with the connected component that contains the depot vertex $0$. 
Let $W_0,\ldots,W_h$ denote these components. If $|V|$ is even, then every component $W_s$ induces on $M_P\cup M_D$ an elementary cycle of even length. Otherwise, $W_s$ induces an elementary chain on an odd number of vertices for a single index $s\in\{0,\ldots,h\}$.
We assume {\em w.l.o.g} that $0\in W_0$; thus vertex set $W_0$ induces either a cycle $(v_{0,0}=0,v_{0,1},\ldots,v_{0,m_0},0)$, or a chain which consists of a cycle $(v_{0,0}=0,v_{0,1},\ldots,v_{0,m_0},0)$ minus some edge $(v_{0,j},v_{0,j+1})$ for a single index $j\in \{0,\ldots,m_0\}$ (index $j+1$ is taken modulo $m_0+1$). %
For $W_0$, the heuristic inserts the sequences
$$\begin{array}{lll}
					&(v_{0,1},\ldots,v_{0,\lceil m_0/2\rceil}) 			&\textrm{at the beginning of }P_1\\
	\textrm{and}		&(v_{0,m_0},\ldots,v_{0,\lceil m_0/2\rceil+1})		&\textrm{at the beginning of }P_2.
\end{array}$$
Any other component $W_s$ induces either a cycle $(v_{s,1},\ldots,v_{s,m_s},v_{s,1})$ or a chain $(v_{s,1},\ldots,v_{s,m_s})$, where the nodes $v_{s,1},\ldots,v_{s,m_s}$ all belong to $V\backslash\{0\}$. The heuristic inserts the sequences
$$\begin{array}{lll}
					&(v_{s,1},\ldots,v_{s,\lfloor m_s/2\rfloor}) 		&\textrm{at the end of }P_1\\
	\textrm{and}		&(v_{s,m_s},\ldots,v_{s,\lfloor m_s/2\rfloor+1})		&\textrm{at the end of }P_2.
\end{array}$$
Figure \ref{fig-apx-P} provides an illustration of the obtained loading plan. 
In any case, the obtained approximate packing $\mathcal{P}$ satisfies the capacity constraints: row $1$ receives plus one vertex {\em vs.} row $2$ when loading vertices from $W_0$ if $W_0$ induces a cycle, whereas it receives minus one vertex {\em vs.} row $2$ when loading vertices from $W_s$ for some index $s\in\{1,\ldots,h\}$ if $|V|$ is odd and $W_s$ is the single component that induces a chain.

Third, it computes the best pair $(T_P^*(\mathcal{P}),T_D^*(\mathcal{P}))$ of pickup and delivery tours with respect to $\mathcal{P}$ and opt. According to Proposition \ref{prop-TSP}, this last step requires a $\mathcal{O}({|V|}^2)$ computation time. 
The overall complexity of Algorithm \ref{algo-apx} is therefore polynomial.

\subsection{Approximation analysis}\label{sec-apx-apx}

\begin{theorem}\label{thm-apx}
Algorithm \ref{algo-apx} provides within polynomial time a standard approximation guarantee of
\begin{enumerate}
	\item[(i)] \label{it-apx-max} $1/2 -1/(2|V|)$ for $\mathsf{Max\,DTSPMS}$,
	\item[(ii)] \label{it-apx-max_ab} $1/2 \times (1 +a/b) -1/(2|V|) \times (1 -a/b)$ for $\mathsf{Max\,DTSPMS-(a,b)}$, 
	\item[(iii)] \label{it-apx-min_ab} $1/2 \times(1 +b/a) +1/(2|V|) \times (b/a -1)$ for $\mathsf{Min\,DTSPMS-(a,b)}$.
\end{enumerate}
Moreover, all these approximation ratios are tight.
\end{theorem}

\begin{figure}[t]
\begin{center}
	\includegraphics{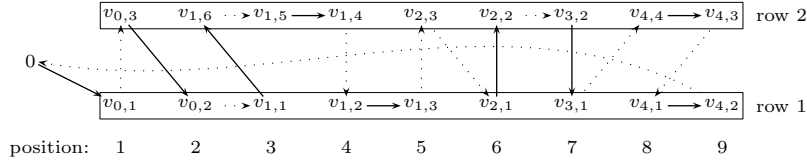}
	\caption{\label{fig-apx-T_P}\small{A feasible pickup tour that uses all the edges in $M_P$. Edges of $M_P$ are depicted in plain lines.}}
\end{center}
\end{figure}

\begin{proof}
Let $APX=d_P(T_P^*(\mathcal{P}))+d_D(T_D^*(\mathcal{P}))$ denote the value of the approximate solution.
By construction, given $\alpha\in\{P,D\}$, $M_\alpha$ is consistent with the approximate loading plan $\mathcal{P}$. Said equivalently, there exist two edge sets $A_P,A_D$ such that $(M_P\cup A_P,M_D\cup A_D)$ defines a feasible pair of pickup and delivery tours with respect to $\mathcal{P}$ (see Figure \ref{fig-apx-T_P} for some illustration).
Since Algorithm \ref{algo-apx} returns the best pair of such tours, the approximate value satisfies:
\begin{align}\label{eq-solref}
	APX	&\succeq d_P(M_P)+d_D(M_D) +d_P(A_P)+d_D(A_D)
\end{align}

If $|V|$ is even, then any tour $T$ on $V$ is the union of two perfect matchings on $V$. Since $M_\alpha, \alpha\in\{P,D\}$ are optimal-weight matchings, we deduce:
\begin{align}\label{eq-opt-odd}
	d_\alpha(M_\alpha)	&\succeq \frac{1}{2}OPT(I_\alpha),	&\alpha\in\{P,D\}
\end{align}
We deduce from relations (\ref{eq-solref}),(\ref{eq-opt-odd}) and (\ref{eq-TSP-STSP-opt}):
\begin{align}\label{eq-pvapx-odd}
    APX	&\succeq \frac{1}{2}OPT(I) +d_P(A_P)+d_D(A_D)
\end{align}
This enables to conclude result (i) for the $\mathsf{Max\,2\,DTSPMS}$, considering $d_P(A_P)\geq 0$ and $d_D(A_D)\geq 0$. For the bivalued case given two real numbers $0<a<b$, considering that $d_P(A_P)+d_D(A_D)$ and $OPT(I)$ express as the sum of respectively $|V|$ and $2|V|$ edge distances, we have:
\begin{align}\nonumber
 	\frac{d_P(A_P)+d_D(A_D)}{OPT(I)}	 \geq \frac{|V|a}{2|V|b} = \frac{a}{2b},
 	&&\frac{d_P(A_P)+d_D(A_D)}{OPT(I)}	 \leq \frac{|V|b}{2|V|a} = \frac{b}{2a}
\end{align}
This leads to (ii) and (iii) for respectively the maximimization and the minimization cases.

\medskip
When $|V|$ is odd, given any Hamiltonian cycle $T$ on $V$ and any edge $e\in T$, the edge set $T\backslash\{e\}$ consists of the union of two near-perfect matchings on $V$. Given $\alpha\in\{P,D\}$ and a tour $T_\alpha^*$ on $V$ that is of optimal with respect to $d_\alpha$ and opt, we denote by $e_\alpha$ an arc of $T_\alpha^*$ having maximum distance $d_\alpha$ if the goal is to minimize, and minimum distance $d_\alpha$ otherwise. $M_\alpha$ and $e_\alpha$ therefore satisfy $2 d_\alpha(M_\alpha) \succeq (d_\alpha(T_\alpha^*)-d_\alpha(e_\alpha))$ and $d_\alpha(e_\alpha) \preceq d_\alpha(T_\alpha^*)/|V|$. As a consequence,
\begin{align}\label{eq-facts-opt-even}
	d_\alpha(M_\alpha)	&\succeq \frac{1}{2}\left(1-\frac{1}{|V|}\right)OPT(I_\alpha),	&\alpha\in\{P,D\}
\end{align}
We derive (i) for the $\mathsf{Max\,2\,DTSPMS}$, considering again $d_\alpha(A_\alpha)\geq 0,\alpha\in\{P,D\}$. As for results (ii) and (iii) for the $\mathsf{2\,DTSPMS-(a,b)}$, similarly to the even case, we observe:
\begin{align}
 	\frac{d_P(A_P)+d_D(A_D)}{OPT(I)}		&\geq \frac{(|V|+1)a}{2|V|b} = \frac{a}{2b}\left(1+\frac{1}{|V|}\right)\nonumber\\
 	\frac{d_P(A_P)+d_D(A_D)}{OPT(I)}		&\leq \frac{(|V|+1)b}{2|V|a} = \frac{b}{2a}\left(1+\frac{1}{|V|}\right)\nonumber
\end{align}
We deduce from the above relations together with relations (\ref{eq-solref}) and (\ref{eq-facts-opt-even}) a standard approximation ratio of 
$1/2\times (1-1/|V|) +a/(2b)\times (1+1/|V|)=1/2 +a/(2b) -o(1)$ when the goal is to maximize, of 
$1/2\times (1-1/|V|) +b/(2a)\times (1+1/|V|)=1/2 +b/(2a)+o(1)$ when the goal is to minimize.

\medskip
In order to establish the tightness of the analysis, we consider bivaluated instances $I(\lambda,\mu,n)$, $n\in\mathbb{N}^*$, $\lambda\neq \mu\in \mathbb{Q}^+$ of the Symmetric $\mathsf{2\,DTSPMS}$. Given an integer $n$ and two reals $\lambda,\mu$, $I(\lambda,\mu,n)=(4n,2n,d_P,d_D,\mathrm{opt})$ where $\mathrm{opt}=\max$ if $\lambda>\mu$ and $\min$ otherwise, and distances $d_P,d_D$ take value $\mu$ on all edges, but along the cycle $(0,1,\ldots,4n,0)$.
We denote by $V_n =\{0, 1 ,\ldots, 4n\}$ the vertex set in $I(\lambda,\mu,n)$. For this instance, the pair $(T_P=(0,1,\ldots,4n,0),T_D=T_P^-)$ of pickup and delivery tours is optimal, with value 
\begin{align}\label{eq-2_tight-OPT}
	OPT(I(\lambda,\mu,n))	&= 2|V_n|\lambda
\end{align}
Now assume that when running Algorithm \ref{algo-apx} on $I(\lambda,\mu,n)$, both $M_P$ and $M_D$ pick edges $\{2i-1,2i\}$, $i\in\{1,\ldots,2n\}$. Additionnally assume that the loading plan $\mathcal{P}=(P_1,P_2)$ built from $M_P,M_D$ is the following:
$$\begin{array}{rclr}
	P_1	&= &(4,8,\ldots,4n,		&1,5,\ldots,4n-3)\\
	P_2	&= &(3,7,\ldots,4n-1,	&2,6,\ldots,4n-2)
\end{array}$$
Observe that the edges of the cycle $(0,1,\ldots,4n,0)$ that are consistent with $\mathcal{P}$ precisely are the edges of $M_P =M_D$. Accordingly, Algorithm \ref{algo-apx} returns a solution with value 
\begin{align}\label{eq-2_tight-APX}
	APX(I(\lambda,\mu,n))	
	    &= (|V_n| -1)\lambda +(|V_n| +1)\mu 
\end{align}
Combining (\ref{eq-2_tight-OPT}) and (\ref{eq-2_tight-APX}), one gets:
$$\begin{array}{rl}\displaystyle
    \frac{APX(I(\lambda,\mu,n))}{OPT(I(\lambda,\mu,n))} 
		&=\frac{1}{2}\left(1 -\frac{1}{|V_n|}\right) +\frac{1}{2}\left(1 +\frac{1}{|V_n|}\right)\frac{\mu}{\lambda}\\
		&=\left\{\begin{array}{rl}
			\frac{1}{2}\left(1 -\frac{1}{|V_n|}\right)		&\textrm{if $(\lambda,\mu)=(1,0)$}\\
		    \frac{1}{2} \left(1 +\frac{a}{b}\right) -\frac{1}{2|V_n|}\left(1 -\frac{a}{b}\right)	&\textrm{if $(\lambda,\mu)=(b,a)$}\\
	    \frac{1}{2} \left(1 +\frac{b}{a}\right) +\frac{1}{2|V_n|}\left(\frac{b}{a} -1\right)		&\textrm{if $(\lambda,\mu)=(a,b)$}
		\end{array}\right.
\end{array}$$
Families $I(1,0,n),n\in\mathbb{N}^*$, $I(b,a,n),n\in\mathbb{N}^*$ and $I(a,b,n),n\in\mathbb{N}^*$ thus establish the tightness of the analysis for respectively $\mathsf{Max\,DTSPMS}$, $\mathsf{Max\,DTSPMS-(a,b)}$ and $\mathsf{Min\,DTSPMS-(a,b)}$.
\end{proof}


\section{Differential approximation results}\label{sec-apx-dapx} 

In this section, we provide approximation results for the differential approximation ratio, which offers a complementary view of approximation vs the standard ratio, as we shall see further.
The differential ratio is the ratio of $|APX(I)-WOR(I)|$ by the instance diameter $|OPT(I)-WOR(I)|$, where $WOR(I)$ is the value of a worst solution. In that differential framework, $\mathcal{A}$ is said to be $\rho$--approximate for some $\rho:\mathcal{I}_\Pi\rightarrow ]0,1]$ {\em iff} 
\begin{align}\nonumber
	\frac{APX(I)-WOR(I)}{OPT(I)-WOR(I)}	&\geq \rho(I),	&I\in \mathcal{I}_\Pi
\end{align}
{\em i.e.,} $APX(I) \geq \rho(I) OPT(I) +(1-\rho(I)) WOR(I)$, $I\in \mathcal{I}_\Pi$ if the goal is to maximize, 
$APX(I) \leq \rho(I) OPT(I) + (1-\rho(I)) WOR(I)$, $I\in \mathcal{I}_\Pi$ otherwise.

As for standard approximation, $\Pi$ is said to be approximable within factor $\rho$ with the differential ratio {\em iff} it admits a polynomial time $\rho$--approximation algorithm. 
For more insights about the differential approximation measure, we invite the reader to refer to \cite{DP96}.

Many differential approximation results have been provided for routing and $\mathsf{TSP}$ related problems \cite{M02,MPT03a,MPT03b,N09,MT08,BHM05,EM08,HK01}. 
For example the symmetric $\mathsf{TSP}$ is approximable within differential factor $3/4-\varepsilon$ \cite{EM08}. 


\subsection{Properties of differential vs standard ratio for the TSP }\label{sec-intro-apx_pty}

 The $\mathsf{TSP}$ has the interesting property that the minimization, maximization and metric cases are all equivalent as regards to differential approximation, which is illustrated in what follows.

The restriction of the $\mathsf{TSP}$ to metric instances is denoted by $\mathsf{Metric\,TSP}$. Furthermore, $\mathsf{Min\,TSP}$ refers to the $\mathsf{TSP}$ where the goal is to minimize, whereas $\mathsf{Max\,TSP}$ refers to the $\mathsf{TSP}$ where the goal is to maximize. 
Let $I=(V,d)$ be an instance of the Symmetric $\mathsf{Min\,TSP}$, characterized by:
$$\begin{array}{lrl}
	 &OPT(I)		&=\min\{d(T)\,|\,T \in {\cal T}_V\}
\end{array}$$
where ${\cal T}_V$ denotes the set of Hamiltonian tours on $V$. Let us note $d_{max}= \max_{i,j\in V:i\neq j} \{d(i,j)\}$ and $d_{min}= \min_{i,j\in V:i\neq j} \{d(i,j)\}$ the maximum and the minimum distances between any pair of nodes, and consider the two instances $I_1,I_2$ defined as:
$$\begin{array}{ll}
	OPT(I_1) =\max\{d_1(T)\,|\,T\in {\cal T}_V\}		&\textrm{where }d_1 =d_{max} -d\\
	OPT(I_2)	 =\min\{d_2(T)\,|\,T\in {\cal T}_V\}	&\textrm{where }d_2	=d +d_{max} -2d_{min}
\end{array}$$
Distances $d_1,d_2$ are non-negative. Furthermore, one can easily check that $d_2$ satisfies the triangle inequalities.
$I_1$ and $I_2$ therefore are instances of the $\mathsf{Max\,TSP}$ and of the $\mathsf{Min\,Metric\,TSP}$, respectively, that can be equivalently expressed as:
$$\begin{array}{rll}
	OPT(I_1)		&=|V|d_{max} 			-\min\{d(T)\,|\,T \in {\cal T}_V\}\\
	OPT(I_2)		&=|V|(d_{max} -2d_{min})	+\min\{d(T)\,|\,T\in {\cal T}_V\}
\end{array}$$
Hence, the three instances $I,I_1,I_2$ of the $\mathsf{TSP}$ correspond to the same optimization problem, up to an affine transformation of their objective function. Accordingly, these  instances are equivalent to differentially approximate. Indeed, observe that for all $T \in {\cal T}_V$, we have:
$$\begin{array}{rll}
	\displaystyle\frac{d(T)- WOR(I)}{OPT(I)-WOR(I)}
		&\displaystyle=\frac{d_1(T)- WOR(I_1)}{OPT(I_1)- WOR(I_1)}
		&\displaystyle=\frac{d_2(T)- WOR(I_2)}{OPT(I_2)- WOR(I_2)}
\end{array}$$

Hence, and in contrast with the standard approximation framework, the $\mathsf{Min\,TSP}$, the $\mathsf{Max\,TSP}$ and their restriction to the metric case are strictly equivalent to differentially approximate. 
The symmetric case of these problems notably all are approximable within a differential factor of $3/4-\varepsilon, \varepsilon>0$, \cite{EM08}. 

Using similar arguments, $\mathsf{Min\,DTSPMS}$, $\mathsf{Max\,DTSPMS}$, $\mathsf{Min\,Metric\,DTSPMS}$ and $\mathsf{Max\,Metric\,DTSPMS}$ are equivalent with respect to their differential approximability. 

\subsection{Differential approximation of the general $\mathsf{DTSPMS}$ }  \label{sec:DIFF}

In Section \ref{sec-tsp}, we derived standard approximation results for the $\mathsf{DTSPMS}$ from connections between the optimal values of a given instance $I$ of the $\mathsf{DTSPMS}$ and of instances $I_P$, $I_D$ and $I_\Sigma$ of the $\mathsf{TSP}$. Such connections between the extremal values on $I$ and $I_\Sigma$ similarly allow to derive differential approximation results for $\mathsf{DTSPMS}$ from differential approximation results for $\mathsf{TSP}$. 

First, symmetrically to (\ref{eq-TSP-STSP-opt}), the worst solution values on instances $I_P, I_D, I_\Sigma$ of the $\mathsf{TSP}$ and $I$ of the $\mathsf{DTSPMS}$ obviously satisfy:
\begin{align}\label{eq-TSP-STSP-wor}
	WOR(I_P) +WOR(I_D)\ \preceq WOR(I)\ \preceq WOR(I_\Sigma)
\end{align}

Now let $(T_P, T_D)$ refer to an optimal solution of $I$. On the one hand, $T_P$ and $T_D^-$ both are feasible solutions of $I_\Sigma$. Therefore, we have:
$$\begin{array}{rl}
	OPT(I_\Sigma)
	   &\succeq 1/2 \times \left(d_P(T_P) +d_D(T_P^-) +d_P(T_D^-) +d_D(T_D)\right)\\
	   &\succeq 1/2 \times \left(OPT(I) +d_P(T_D^-) +d_D(T_P^-) \right)    
\end{array}$$ 
On the other hand, $(T_D^-, T_P^-)$ is a feasible pair of pickup and delivery tour on $I$. Accordingly, we have $d_P(T_D^-) +d_D(T_P^-) \succeq WOR(I)$. 
The following Proposition thus holds:
\begin{lemma}\label{lem-dapx} Given any instance $I$ of the $\mathsf{DTSPMS}$, $I$ and its related instance $I_\Sigma$ of the $\mathsf{TSP}$ satisfy:
	\begin{align}\label{eq-TSP-STSP-wor+opt}
	    OPT(I_\Sigma) &\succeq \left(WOR(I)) +OPT(I)\right)/2
	\end{align}
\end{lemma}

Relation (\ref{eq-TSP-STSP-wor+opt}) indicates that the optimal value of $I_\Sigma$ provides a 1/2-differential approximation of $OPT(I)$. 
It also yields a rather simple differential approximation preserving reduction from $\mathsf{DTSPMS}$ to $\mathsf{TSP}$.

\begin{proposition}\label{prop-reduc-stsp2tsp-DIFF}
	The (Symmetric) $\mathsf{DTSPMS}$ reduces to the (Symmetric) $\mathsf{TSP}$ by means of a polynomial time reduction that maps $\rho$--differential approximate solutions of the $\mathsf{TSP}$ onto solutions of the $\mathsf{DTSPMS}$ with a differential approximation guarantee of $\rho/2$.
\end{proposition}

\begin{proof}
Let $I$ be an instance of $\mathsf{DTSPMS}$. Given any tour $T$ on $V$, $(T, T^-)$ is a feasible pair of pickup and delivery tours on $I$, with value $d_\Sigma(T)$. In particular if $T$ is $\rho$--approximate for the $\mathsf{TSP}$ on $I_\Sigma$, then we have:
$$\begin{array}{rll}
	d_\Sigma(T)
	   &\succeq     \rho\, OPT(I_\Sigma) +(1 -\rho) WOR(I_\Sigma)        &\\
	   &\succeq     \rho \left(OPT(I) +WOR(I)\right)/2 	+(1 -\rho) WOR(I)
	           &\text{using (\ref{eq-TSP-STSP-wor}) \& (\ref{eq-TSP-STSP-wor+opt})}\\
	   &=           \rho/2 \times OPT(I) + (1 -\rho/2)  WOR(I)    
\end{array}$$ 
Solution $(T, T^-)$ therefore is $\rho/2$--approximate on $I$. 
\end{proof}

The theorem below is a straightforward consequence of Proposition \ref{prop-reduc-stsp2tsp-DIFF} and the result of \cite{EM08}.
\begin{theorem}\label{thm-dapx-from-tsp}
	The Symmetric $\mathsf{DTSPMS}$ is approximable within differential ratio $3/4 -\varepsilon$, $\varepsilon >0$.
\end{theorem}

\subsection{Differential approximation of the $\mathsf{2\,DTSPMS}$}

Some adaptation of the heuristic of Section \ref{sec-apx} enables to reach a differential approximation ratio of $1/2-o(1)$ for the $\mathsf{2\,DTSPMS}$.
In the proposed heuristic, the computation of optimal matchings brings {\em ``one half''} of the optimal value $OPT(I)$, which allows to establish a standard  approximation guarantee of $1/2 -o(1)$ for the maximization case. Obtaining such a guarantee with respect to the differential approximation measure additionally requires the comparison of the remaining part of the approximate solution -- namely, completions $A_P$ and $A_D$ of matchings $M_P$ and $M_D$ -- to the worst solution value $WOR(I)$. 
This comparison to the worst solution value captures the specificity of differential approximation, and may make it hard to establish differential approximation guarantees.  

\begin{figure}[t]
\begin{center}
	\includegraphics{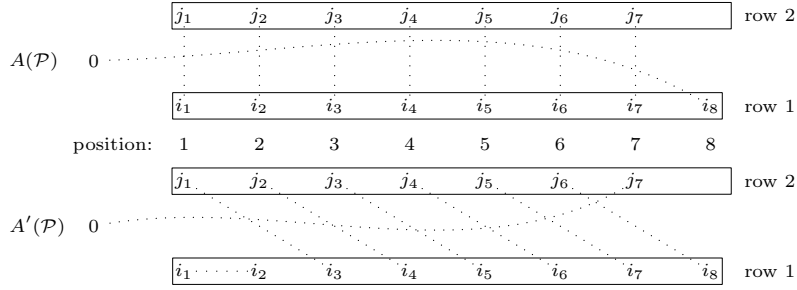}
	\caption{\label{fig-dapx-N}\small{The perfect matchings $A(\mathcal{P})$ and $A'(\mathcal{P})$ given a loading plan $\mathcal{P}$ (both matchings are drawn in dotted lines).}}
\end{center}
\end{figure}

\begin{figure}[t]
\begin{center}
	\includegraphics{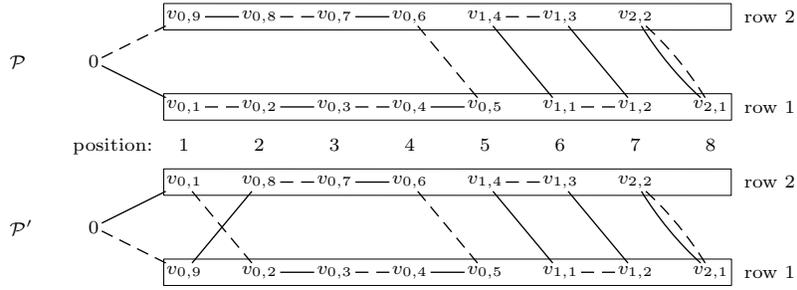}
	\caption{\label{fig-dapx-P}\small{Loading plans $\mathcal{P} =\mathcal{P}(M_P, M_D)$ and $\mathcal{P}' =\mathcal{P}'(M_P, M_D)$ given two  perfect matchings $M_P$ (in plain lines) and $M_D$ (in dashed lines).}}
\end{center}
\end{figure}

Consider an instance $I$ of the $\mathsf{2\,DTSPMS}$ where $|V| =2\nu +2$ is even (we will speak later of the case when $|V|$ is odd). 
We seek perfect matchings that complement the matchings $M_P$ and $M_D$. 
With a given balanced loading plan $\mathcal{P}=((i_1,\ldots,i_{\nu +1}),(j_1,\ldots,j_\nu))$ of $V\backslash\{0\}$, we associate the two perfect matchings
$$\begin{array}{rl}	
	A(\mathcal{P})	&= \left\{(i_p,j_p)\,|\,p=1,\ldots,\nu\right\} \cup\left\{(i_{\nu +1},0)\right\}\\
	A'(\mathcal{P})	&= \left\{(i_1,i_2)\right\} \cup\left\{(i_p,j_{p-2})\,|\,p=3,\ldots,\nu +1\right\} \cup\left\{(j_\nu,0)\right\}
\end{array}$$
on $V$. These matching are depicted in Figure \ref{fig-dapx-N} in case when $\nu =7$. Furthermore, we denote by $\mathcal{P}'$ the loading plan obtained from $\mathcal{P}$ when exchanging the storage of the two nodes that are loaded at position $1$ of rows $1$ and $2$.

Thereafter, we consider a pair $(\mathcal{P},\mathcal{P}')$ of loading plans where $\mathcal{P} =\mathcal{P}(M_P, M_D)$ refers to the approximate loading plan of Section \ref{sec-apx}. 
Figure \ref{fig-dapx-P} depicts the loading plans $\mathcal{P}$ and $\mathcal{P}'$ given two perfect matchings $M_P$ and $M_D$. 
The following Lemma holds:

\begin{lemma}\label{lem-dapx-even}
Let $A=A(\mathcal{P})=A(\mathcal{P}')$. 
Then,
\begin{enumerate} 
	\item[(i)] $M_P \cup A$ is a feasible pickup tour with respect to $\mathcal{P}$ and $\mathcal{P}'$, 
			and $M_D \cup A$ is a feasible delivery tour with respect to $\mathcal{P}$ and $\mathcal{P}'$.
	\item[(ii)] For $\alpha\in\{P,D\}$, if $M_\alpha$ links the depot to the vertex which is loaded in $\mathcal{P}$ at position 1 of row 1, then $M_\alpha\cup A'(\mathcal{P})$ is a feasible tour with respect to $\mathcal{P}$; symmetrically, if $M_\alpha$ links vertex 0 to the first vertex in row 1 of $\mathcal{P}'$, then $M_\alpha\cup A'(\mathcal{P}')$ is a feasible tour with respect to $\mathcal{P}'$. 
	\item[(iii)] $(A\cup A'(\mathcal{P}), A\cup A'(\mathcal{P}'))$ is a feasible pair of pickup and delivery tours on $V$ for the $\mathsf{2\,DTSPMS}$ with tight capacity.
\end{enumerate} 
\end{lemma}

\begin{figure}[t]
\begin{center}
	\includegraphics{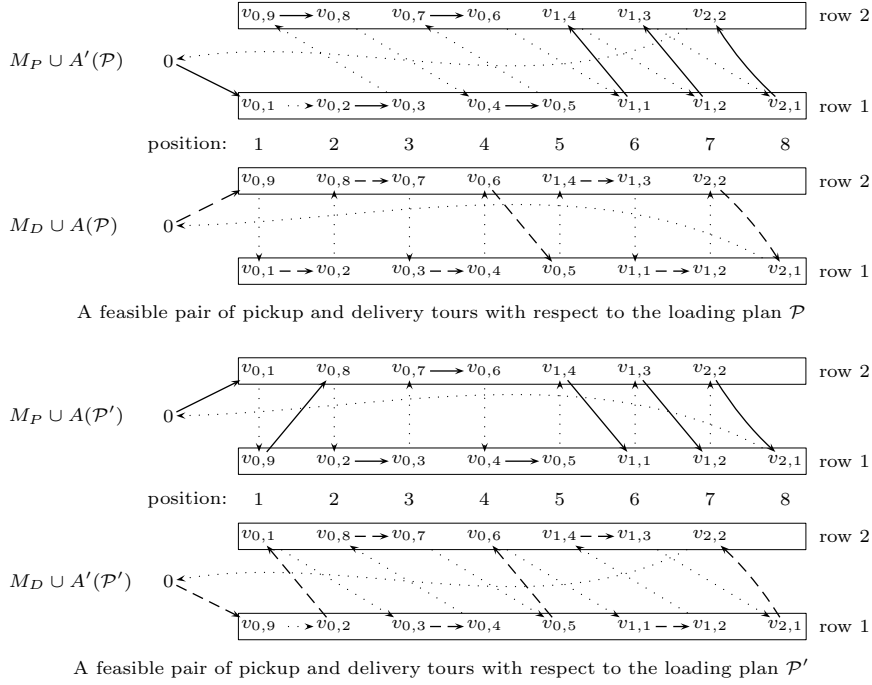}
	\caption{\label{fig-dapx-N+P}\small{The approximate solutions $(\mathcal{P}, M_P\cup A'(\mathcal{P}), M_D\cup A)$ and $(\mathcal{P}', M_P\cup A), M_D\cup A'(\mathcal{P}'))$: $M_P$ is depicted in plain lines, $M_D$ is depicted in dashed lines, $A$, $A'(\mathcal{P})$ and $A'(\mathcal{P}')$ are depicted in dotted lines.}}
\end{center}
\end{figure}

\begin{figure}[t]
\begin{center}
	\includegraphics{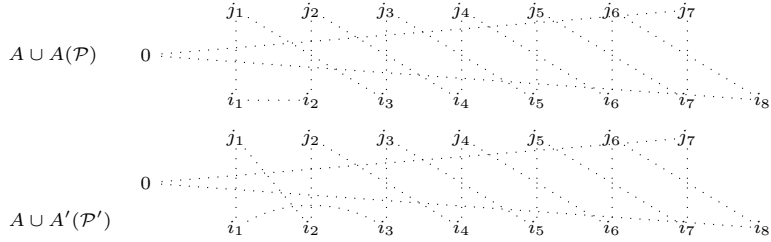}
	\caption{\label{fig-dapx-N+Np}\small{The two tours $A\cup A'(\mathcal{P})$ and $A\cup A'(\mathcal{P}')$ given a loading plan $\mathcal{P}=\left((i_1,\ldots,i_{\nu +1}),(j_1,\ldots,j_\nu)\right)$}}
\end{center}
\end{figure}

Figures \ref{fig-dapx-N+P} and \ref{fig-dapx-N+Np} propose some illustration of these facts.

\begin{proof}
(i) and (ii) Let $\mathcal{Q}=\left((i_1,\ldots,i_{\nu +1}),(j_1,\ldots,j_\nu)\right)\in\{\mathcal{P},\mathcal{P}'\}$.
 %
First consider completion $A$. By construction of $\mathcal{P}$ and $\mathcal{P}'$, each of the two perfect matchings $M_P,M_D$ connects 
	vertex $0$ to either $i_1$ or $j_1$, 
	a single vertex in $\{i_p,j_p\}$ to a single vertex in $\{i_{p+1},j_{p+1}\}$ for each position $p\in\{1,\ldots,\nu -1\}$, 
	and either $i_\nu$ or $j_\nu$ to vertex $i_{\nu +1}$.
$M_P\cup A$ and $M_D\cup A$ therefore both define Hamiltonian cycles on $V$, and these cycles induce feasible pickup and delivery tours with respect to $\mathcal{Q}$. 

Similarly for $A'(\mathcal{Q})$, for all $\alpha\in\{P,D\}$ such that $(0,i_1)\in M_\alpha$, $M_\alpha$ connects:
	$0$ to $i_1$;
	$i_2$ to $\{i_3,j_1\}$;  
	$\{i_p,j_{p-2}\}$ to $\{i_{p+1},j_{p-1}\}$ for every position $p\in\{3,\ldots,\nu\}$,
	and $\{i_{\nu +1},j_{\nu -1}\}$ to $j_\nu$.
We deduce that $M_\alpha\cup A'(\mathcal{Q})$ induces a feasible tour with respect to $\mathcal{Q}$ provided that $(0,i_1)\in M_\alpha$. 

(iii) Let $\mathcal{P} =\left((i_1,\ldots,i_{\nu +1}),(j_1,\ldots,j_\nu)\right)$. 
By definition of $A$ and $A'(\mathcal{P})$, $A\cup A'(\mathcal{P})$ can be viewed as the tour either 
$$\begin{array}{rl}
    &\left(0, j_\nu, i_\nu, j_{\nu -2}, i_{\nu -2} ,\ldots, j_2, i_2,
    i_1, j_1,
    i_3, j_3, i_5, j_5 ,\ldots, i_{\nu -1}, j_{\nu -1}, i_{\nu +1}, 0\right)
\\\text{or}
    &\left(0, j_\nu, i_\nu, j_{\nu -2}, i_{\nu -2} ,\ldots, j_3, i_3,
    j_1, i_1,
    i_2, j_2, i_4, j_4 ,\ldots, i_{\nu -1}, j_{\nu -1}, i_{\nu +1}, 0\right)
\end{array}$$
on $V$, depending on $\nu \mod 2$. Let $T$ refer to this tour.
Furthermore, by definition of $\mathcal{P}'$, $A\cup A'(\mathcal{P}')$ can be obtained from $A\cup A'(\mathcal{P})$ by substituing with the two edges $(i_1, i_2)$ and $(j_1, i_3)$ the edges $(j_1, i_2)$ and $(i_1, i_3)$. 
Therefore, $A\cup A'(\mathcal{P}')$ induces on $V$ a tour $T'$ which just the same as $T$, but swapping the two vertices $i_1$ and $j_1$. We deduce that the pair $(T', T^-)$ of tours defines a feasible pair of pickup and delivery tours on $V$, considering {\em e.g.} the loading plans 
$$\begin{array}{rl}
    &\left(
        (j_\nu, i_\nu, j_{\nu -2}, i_{\nu -2} ,\ldots, j_2, i_2, i_1), 
        (j_1, i_3, j_3, i_5, j_5 ,\ldots, i_{\nu -1}, j_{\nu -1}, i_{\nu +1})
    \right)
\\\text{and}
    &\left(
        (j_\nu, i_\nu, j_{\nu -2}, i_{\nu -2} ,\ldots, j_3, i_3, j_1), 
        (i_1, i_2, j_2, i_4, j_4 ,\ldots, i_{\nu -1}, j_{\nu -1}, i_{\nu +1})
    \right)
\end{array}$$
for respectively the even and the odd cases.
\end{proof} 

\begin{algorithm}\label{algo-dapx-even}
\caption{DAPX\_2DTSPMS\_EVEN}
\KwIn{A vertex set $V =\{0,\ldots,n\}$ where $n$ is odd, two symmetric distance functions $d_P,d_D: V^2\rightarrow\mathbb{Q}^+$, an optimization goal $\mathrm{opt}$}
\KwOut{A balanced $2$-rows loading plan $\mathcal{P}$ of $V\backslash\{0\}$ and an optimal pair $\left(T_P^*(\mathcal{P}),T_D^*(\mathcal{P})\right)$ of pickup and delivery tours on $V$ with respect to $d_P,d_D,\mathrm{opt}$ and $\mathcal{P}$}

\BlankLine
$\left(\mathcal{P},T_P,T_D\right)\longleftarrow$ APX\_2DTSPMS\_EVEN$(n,d_P,d_P,\mathrm{opt})$\;

\BlankLine
$\mathcal{P}'\longleftarrow\mathcal{P}$\;
Exchange in $\mathcal{P}'$ the nodes that are stored at position $1$ in rows $1$ and $2$\;
\ForEach{$\alpha$ in $\{P,D\}$}
{
	Compute a tour $T'_\alpha$ on $V$ that is consistent with $\mathcal{P}'$ and optimal with respect to $d_\alpha$ and $\mathrm{opt}$\;
}

\BlankLine
\If{$d_P(T_P)+d_D(T_D)\succeq d_P(T'_P)+d_D(T'_D)$}
	{\Return $\left(\mathcal{P},T_P,T_D\right)$\;}
\Else
	{\Return $\left(\mathcal{P}',T'_P,T'_D\right)$\;}
\end{algorithm}

\begin{theorem}\label{thm-dapx-even}
	The Symmetric $\mathsf{2\,DTSPMS}$ is approximable within a differential factor of $1/2 -o(1)$. 
\end{theorem}

\begin{proof}
In case when $|V|$ is even, we show that Algorithm \ref{algo-dapx-even} provides a $1/2$--differential approximation for the Symmetric $\mathsf{2\,DTSPMS}$, {\em i.e.} for any instance $I$ Algorithm \ref{algo-dapx-even} returns a solution with value $APX \succeq  OPT(I)/2 +WOR(I)/2$. 
We assume {\em without loss of generality} that the goal on $I$ is to maximize. 
Let $m_0$ denote the number of vertices that lie in $M_P\cup M_D$ on the cycle that contains 0. We separate the proof in two parts, depending on whether $m_0=2$ or $m_0 \geq 4$. 

Let $\mathcal{P} =\left((i_1,\ldots,i_{\nu +1}),(j_1,\ldots,j_\nu)\right)$. 
When $m_0=2$, the perfect matchings $M_P$ and $M_D$ both contain edge $(0, i_1)$. It thus follows from Lemma \ref{lem-dapx-even} that $M_P\cup A$ and $M_P\cup A'(\mathcal{P})$ on the one hand, $M_D\cup A$ and $M_D\cup A'(\mathcal{P})$ on the other hand, are feasible pickup and delivery tours with respect to $\mathcal{P}$. 
Since Algorithm \ref{algo-dapx-even} returns a best pair of pickup and delivery tours with respect to $\mathcal{P}$ or $\mathcal{P}'$, we deduce that the value $APX$ of the solution returned by the Algorithm satisfies:
$$\begin{array}{rl}
	APX &\geq \max\left\{
	        d_P(M_P \cup A) +d_D(M_D \cup A),\right.\\ 
	   &\left.\hspace*{2cm}     
	           d_P(M_P \cup A'(\mathcal{P})) +d_D(M_D \cup A'(\mathcal{P}))
	    \right\}\\
	    &\geq d_P(M_P) +d_D(M_D)
	            + d_\Sigma\left(A \cup A'(\mathcal{P})\right)/2
\end{array}$$
We already know that quantity $d_P(M_P) +d_D(M_D)$ is bounded below by $OPT(I)/2$. Now, since $A\cup A'(\mathcal{P})$ is a Hamiltonian tour on $V$, we also have $d_\Sigma(A\cup A'(\mathcal{P})) \geq WOR(I)$. This  concludes the proof for the case when $m_0 =2$.

When $m_0\geq 4$, either $(0,i_1)\in M_P$ and $(0,j_1)\in M_D$, or $(0,j_1)\in M_P$ and $(0,i_1)\in M_D$. Assume {\em w.l.o.g.} that the former occurs. Lemma \ref{lem-dapx-even} in this case ensures that 
    $(\mathcal{P}, M_P\cup A'(\mathcal{P}),M_D\cup A)$ and $(\mathcal{P}', M_P\cup A,M_D\cup A'(\mathcal{P}'))$ are feasible solutions on $I$.  
Similarly to the preceding case, we deduce from the fact that Algorithm \ref{algo-dapx-even} returns a best pair of tours with respect to $\mathcal{P}$ or $\mathcal{P}'$ that we have:
$$\begin{array}{rl}
	APX &\geq \max\left\{
	        d_P\left(M_P \cup A'(\mathcal{P})\right) +d_D(M_D \cup A),\right.\\ 
	   &\left.\hspace*{2cm}     
	        d_P(M_P \cup A) +d_D\left(M_D \cup A'(\mathcal{P}')\right)
	    \right\}\\
	    &\geq OPT(I)/2
	            +\left(
	                d_P\left(A \cup A'(\mathcal{P})\right)
	               +d_D\left(A \cup A'(\mathcal{P}')\right)
	           \right)/2
\end{array}$$
Now we know from Lemma \ref{lem-dapx-even} that $(A\cup A'(\mathcal{P}),A\cup A'(\mathcal{P}'))$ defines a feasible pair of pickup and delivery tours, which concludes the proof.

In case when $|V|$ is odd, the algorithm mostly consists in computing a loading plan for each $x\in V$, each based on the computation of a pair $(M^x_P,M^x_D)$ of optimal perfect matchings on $V\backslash\{x\}$. 
Since the proof is technical and brings no new insights on the problem, we put it in a separated appendix. 
\end{proof}


\section{Conclusion}\label{sec-conc}

We have provided many approximation results for the Double TSP with Multiple Stacks or its restriction with two stacks, for several kinds of distances. Among them, $\mathsf{Min\,Metric\,k\,DTSPMS}$ with tight capacities is approximable within standard factor $(3/2)k$, 
	whereas $\mathsf{2\,DTSPMS}$ is approximable within differential factor $1/2-o(1)$. Also,  $\mathsf{Max\,2\,DTSPMS}$, $\mathsf{Max\,2\,DTSPMS-(1,2)}$ and 
	$\mathsf{Min\,2\,DTSPMS-(1,2)}$ with tight capacities are approximable within standard factor $1/2-o(1)$, $3/4-o(1)$ and $3/2+o(1)$, respectively.
Most of our positive approximation results on the general problem are obtained from reductions from the TSP. For the problem with two stacks, we designed a dedicated algorithm based on optimal matchings and suitable completions that can be compared to the best and worst tours. The analysis is non trivial and provides interesting approximation results, in both cases of standard and differential approximation. An open problem is to design tailored algorithms for the case with more than two stacks, which could improve the approximation ratios found with TSP reductions. The VRP generalization is also interesting to study, although its complexity would make a real challenge to find approximation results.


\bibliographystyle{siam}
\bibliography{bib_stsp}
 
\appendix
\input{2DTSPMS-dapx.tex}

\end{document}

%% file: 2DTSPMS-dapx.tex
\section{APPENDIX : Differential approximation of the $\mathsf{DTSPMS}$ on an odd number of vertices}

\subsection{The general idea of the proof}

Let $I$ be an instance of the $\mathsf{2\,DTSPMS}$ on a node set $V$ such that $|V|$ is odd. We assume {\em w.l.o.g.} that the goal on $I$ is to maximize.
In what follows, $T_*$ denotes a worst solution on $I_\Sigma$, {\em i.e.}, $T_*$ is a tour of minimum distance $d_\Sigma$. Furthermore, $(\mathcal{P}^*,T_P^*,T_D^*)$ denotes an optimal solution on $I$.
Given $x\in V$, we denote by $V_x$ the vertex set $V\backslash\{x\}$. Morever, given an index $\alpha\in\{P,D\}$, $M^x_\alpha$ refers to a maximal perfect matching on $V_x$ with respect to opt and $d_\alpha$. Finally, given two nodes $i,j\in V_x$,
 $\delta_\alpha^x(i,j)$ refers to the quantity $d_\alpha(i,x)+d_\alpha(x,j)-d_\alpha(i,j)$. By extension, given a tour $T$ on $V$, $\delta_\alpha^x(T)$ refers to $\delta_\alpha^x(i,j)$ for the two vertices $i$ and $j$ that are adjacent to $x$ in $T$.
 %
We make some observation on the extremal values:
\begin{lemma}\label{lem-odd-extr}
$M^x_P,M^x_D,x\in V$, $OPT(I)$, $WOR(I)$ satisfy:
\begin{align}
	\forall x\in V,\ d_P(M^x_P) +d_D(M^x_D) \geq 1/2 \times \left(OPT(I) -\delta^x_P(T_P^*) -\delta^x_D(T_D^*)\right)\label{eq-fact1}\\
	\sum_{x\in V}\left(\delta^x_P(T_P^*)+\delta^x_D(T_D^*) -\delta^x_P(T_*)-\delta^x_D(T_*)\right)	\preceq 4\left(OPT(I)-WOR(I)\right)\label{eq-fact2}
\end{align}
\end{lemma}

\begin{figure}[t]
\begin{center}
	\includegraphics{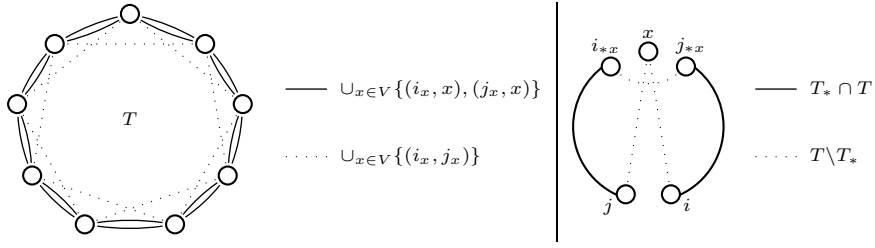}
	\caption{\label{fig-dapx-odd-lemmas}\small{On the left: illustration of inequality (\ref{eq-fact2-pv}) (proof of Lemma \ref{lem-odd-extr}). On the right: illustration of inequality (\ref{eq-fact3-pv2}) (proof of Theorem \ref{thm-dapx-odd})}}
\end{center}
\end{figure}

\begin{proof}
{\em Relation $(\ref{eq-fact1})$}. Given a vertex $x\in V$ and an index $\alpha\in\{P,D\}$, for the two nodes $i,j$ such that $(i,x),(x,j)\in T_\alpha^*$, $T_\alpha^*\backslash\{(i,x),(x,j)\}\cup\{(i,j)\}$ defines a tour on $V_x$. Since $M_\alpha^x$ is a maximal perfect matching on $V_x$, $d_\alpha(M_\alpha^x)$ is at least one half of the value of this tour. Thus we have $2 d_\alpha(M_\alpha^x)\geq d_\alpha(T_\alpha^*) -\delta_\alpha^x(T_\alpha^*), \alpha \in \{P, D\}$, while $\sum_{\alpha =P, D} d_\alpha(T_\alpha^*) =OPT(I)$.

{\em Relation $(\ref{eq-fact2})$}. 
Given a tour $T=(0,v_1,\ldots,v_n,0)$ on $V$, let $T^2$ be the set of arcs $(i,j)$ such that $i$ and $j$ are at distance $2$ in $T$, 
{\em i.e.}, 
$$\begin{array}{rl}
	T^2	&=\{(i,j)\ |\ \exists h,\ (i,h),(h,j)\in T\}
\end{array}$$
If $|V|$ is odd {\em iff} $n$ is even, then $T^2$ is the tour $(0,v_2,v_4,\ldots,v_n,1,3,\ldots,v_{n-1},0)$ on $V$. 
One one hand, for $\alpha\in\{P,D\}$ the quantities $\delta_\alpha^x(T), x\in V$ satisfy 
\begin{align}\label{eq-fact2-pv}
	\sum_{x\in V}\delta_\alpha^x(T) &=2d_\alpha(T)-d_\alpha(T^2)
\end{align}
(see Figure \ref{fig-dapx-odd-lemmas}). Hence, any tour $T$ on $V$ satisfies
$$\begin{array}{lll}
	\sum_{x\in V} \left(\delta^x_P(T) +\delta^x_D(T)\right)
		&=2d_\Sigma(T) -d_\Sigma(T^2)
		&\geq 2WOR(I) -OPT(I)
\end{array}$$
On the other hand, $(T,T^{2-})$ is a feasible pair of pickup and delivery tours: consider {\em e.g.} the loading plan $\left((v_1, v_3, \ldots, v_n),(v_2, v_4, \ldots, v_{n-1})\right)$.
Hence, any feasible pair $(T_P,T_D)$ of pickup and delivery tours satisfies:
$$\begin{array}{rl}
	\multicolumn{2}{l}{\sum_{x\in V}\left(\delta^x_P(T_P)+\delta^x_D(T_D)\right)}\\
	=			&2d_P(T_P) +2d_P(T_D) -d_P(T_P^2) -d_D(T_D^2)\\
	=			&2\left(d_P(T_P) +d_P(T_D)\right) -\left(d_P(T_P^2)+d_D(T_P)\right)	-\left(d_P(T_D)+d_D(T_D^2)\right)\\
				\multicolumn{2}{r}{+\left(d_D(T_P)+d_P(T_D)\right)}\\
	\preceq		&3OPT(I) -2WOR(I)
\end{array}$$
which concludes the lemma.
\end{proof}

\begin{algorithm}\label{algo-dapx-odd}
\caption{Appproximate loading plan for the $\mathsf{DTSPMS}$}
\KwIn{An instance $I=(n,2,\lceil n/2\rceil,d_P,d_D,\mathrm{opt})$ of the $\mathsf{DTSPMS}$ on vertex set $V$ such that $|V|$ is odd}
\KwOut{An approximate loading plan of $V\backslash\{0\}$}

\BlankLine
\For{$x\in V$}
{
	Compute a perfect matching $M^x_\alpha$ on $V\backslash\{x\}$ that is of optimum weight with respect to $d_\alpha$ and $\mathrm{opt}$\;

	\BlankLine
	\tcc{LOADING\_PLAN\_1, LOADING\_PLAN\_2, LOADING\_PLAN\_3, LOADING\_PLAN\_4 return a loading plan of $V\backslash\{0\}$ that admits a pickup tour that contains $M^x_P$ as well as a delivery tour that contains $M^x_D$}

	\If{$x\neq 0$}
	{
    	\tcc{The multi-edge set $M^x_P\cup M^x_D$ consists of $h_x +1$ cycles of even length on vertex set $W_s,s \in \{0 ,\ldots, h_x\}$. The depot vertex belongs to $W_0$.}

		\If{$|W_0|=2$} {
			$\mathcal{P}\longleftarrow$ LOADING\_PLAN\_2($V_x,x,M^x_P,M^x_D,d_P,d_D$)\;
		}
		\Else {
			$\mathcal{P}\longleftarrow$ LOADING\_PLAN\_3($V_x,x,M^x_P,M^x_D,d_P,d_D$)\;
		}
	}
	\Else {
    	\tcc{The multi-edge set $M^0_P\cup M^0_D$ consists of $h_0$ cycles of even length on vertex set $W_s,s \in \{1 ,\ldots, h_0\}$.}

    	\If{$x =0$ and $h_0 \geq 2$} {
    		$\mathcal{P}\longleftarrow$ loading\_plan\_4($V_0,0,M^0_P,M^0_D,d_P,d_D$)\;
    	}
    	\Else {
    		$\mathcal{P}\longleftarrow$ loading\_plan\_5($V_0,0,M^0_P,M^0_D,d_P,d_D$)\;
    	}
    }
}

\BlankLine
\Return $\arg\mathrm{opt}_{\{\mathcal{P}\,|\,x\in V\}} \left\{\sum_{\alpha=P,D} d_\alpha\left(T_\alpha(\mathcal{P})\right)\right\}$;
\end{algorithm}

We adapt the heuristic for the even case to the odd case. The adaptation mainly consists in computing a loading plan per vertex $x\in V$, instead of a single loading plan. 
The algorithm computes for any $x\in V$ a pair $(M^x_P,M^x_D)$ of optimal perfect matchings on $V_x$ and builds a loading plan $\mathcal{P}_x$ that is consistent with both $M^x_P$ and $M^x_D$. Let $APX_x=d_P(T_P^*(\mathcal{P}_x))+d_D(T_D^*(\mathcal{P}_x))$ denote the value of loading plan $\mathcal{P}_x$, $x\in V$; the algorithm then returns the loading plan among $\{\mathcal{P}_x\,|\,x\in V\}$ that achieves the best value $APX_x$. 
Hence, the value $APX$ of the solution returned by the algorithm clearly satisfies:
 %
\begin{align}
	APX \geq \frac{1}{|V|} \sum_{x\in V} APX_x
\label{eq-APX}\end{align}

Before providing the approximability result for the odd case, we need the following lemma. As the proof is long with multiple cases, it is given in a separated section.

\begin{lemma} \label{lem-Fx}
For any node $x\in V$, there exists a subset $F_x$ of $V_x\times V_x$ such that
\begin{enumerate}
	\item[(i)] For all $(i,j)\in F_x$, there exists a pair 
	    $(\mathcal{P}_x, M^x_P\cup N^x_P, M_D\cup N^x_D),
	        (\mathcal{P}_x', M^x_P\cup N'^x_P, M^x_D\cup N'^x_D)$ of feasible solutions on $I$ such that
			    $$\begin{array}{l}
			        \left(\left(N^x_P \cup N'^x_P\right) 
			            \backslash\left\{(x, i), (x, j)\right\} 
			            \cup\left\{(i, j)\right\},\right.\\
			    \left.\qquad\left(N^x_D \cup N'^x_D\right) 
			        \backslash\left\{(x, i), (x, j)\right\} 
			        \cup\left\{(i, j)\right\}
			    \right)
			 \end{array}$$
				is a feasible pair of pickup and delivery tours on $I$
	\item[(ii)] \label{it-Fx-T} $F_x$ intersects all Hamiltonian cycles on $V_x$
\end{enumerate}
\end{lemma}

\begin{theorem}\label{thm-dapx-odd}
	Algorithm \ref{algo-dapx-odd} is a $1/2$--differential approximation for the Symmetric 2DTSPMS, {\em i.e.} for any instance $I$ of the Symmetric $\mathsf{2\,DTSPMS}$ where $|V|$ is odd, Algorithm \ref{algo-dapx-odd} returns a solution with value $APX \geq  OPT(I)/2 +WOR(I)/2$.
\end{theorem}

\begin{proof}
We first show that for all $x\in V$,
\begin{align}\label{eq-fact3}
	APX_x	&\geq d_P(M^x_P) +d_D(M^x_D) +\frac{1}{2}\left(WOR(I) +\delta^x_P(T_*) +\delta^x_D(T_*)\right)
\end{align}

Then we successively deduce from relations (\ref{eq-APX}),(\ref{eq-fact3}),(\ref{eq-fact1}) and (\ref{eq-fact2}) that the approximate value satisfies
$$\begin{array}{rl}
	APX 	&\geq \frac{1}{|V|}\sum_{x\in V} APX_x\\
			&\geq \frac{1}{|V|}\sum_{x\in V} \left(d_P(M^x_P) +d_D(M^x_D) +\frac{1}{2}\left(WOR(I) +\delta^x_P(T_*) +\delta^x_D(T_*)\right)\right)\\
			&\geq \frac{1}{2|V|}\sum_{x\in V} \left(OPT(I) +WOR(I)  -\sum_{\alpha\in\{P,D\}}\left(\delta^x_\alpha(T_\alpha^*) -\delta^x_\alpha(T_*)\right)\right)\\
			&\geq \left(\frac{1}{2}-\frac{2}{|V|}\right) OPT(I) +\left(\frac{1}{2}+\frac{2}{|V|}\right)WOR(I)
\end{array}$$
which ends the proof of the $1/2$-differential ratio.\\

Now, let us prove relation (\ref{eq-fact3}). Consider the edge set $F_x$ of Lemma \ref{lem-Fx} and some edge $(i_x,j_x) \in F_x$ optimizing $\delta_P^x(i,j)+\delta_D^x (i,j)$ over this set. 
According to Lemma \ref{lem-Fx}, there exist two edge sets $N_x,N_x'$ such that $\{i_x,x\},\{x,j_x\}\in N_x\cup N'_x$ and $(\mathcal{P}_x,M^x_P\cup N_x,M^x_D\cup N_x)$, $(\mathcal{P}'_x,M^x_P\cup N'_x,M^x_D\cup N'_x)$ are feasible solutions on $I$. Since Algorithm \ref{algo-dapx-odd} considers both solutions $(\mathcal{P}_x,T_P^*(\mathcal{P}_x),T_D^*(\mathcal{P}_x))$ and $(\mathcal{P}'_x,T_P^*(\mathcal{P}'_x),T_D^*(\mathcal{P}'_x))$, we deduce: 
$$\begin{array}{rl}
	APX_x 	&\geq		\frac{1}{2} \left(2d_P(M^x_P) +2d_D(M^x_D) +d_\Sigma(N_x) +d_\Sigma(N'_x)\right)
\end{array}$$
Moreover, since Lemma \ref{lem-Fx} additionally indicates that $(N_x\cup N'_x)\backslash\{\{x,i_x\},\{x,j_x\}\} \cup\{\{i_x,j_x\}\}$ is a Hamiltonian cycle on $V$, we get that the value of this tour with respect to $d_\Sigma$ is better than $WOR(I)$. Hence,
\begin{align}\label{eq-fact3-pv1}
	APX_x 	&\geq	d_P(M^x_P) +d_D(M^x_D) +\frac{1}{2}WOR(I)+\frac{1}{2}\left(\delta_P(i_x,j_x)+\delta_D(i_x,j_x)\right)
\end{align}

Let $i_{*x}$ and $j_{*x}$ respectively denote the predecessor and the successor of $x$ in $T_*$ and let $T_{*x}$ refer to the tour $T_*\backslash \{(i_{*x},x),(x,j_{*x})\}\cup\{(i_{*x},j_{*x})\}$ on $V_x$. 
On the one hand, since $T_*$ is a tour on $V$ of worst value with respect to $d_\Sigma$ and the optimization goal, $d_\Sigma(T) \geq d_\Sigma(T_*)$ for all tour $T$ on $V$ obtained from $T_*$ by first removing edges $(i_{*x},x),(x,j_{*x})$ as well as some other edge $(i,j)\in T_*$, and then adding the edges $(i,x),(x,j),(i_{*x},j_{*x})$
(see Figure \ref{fig-dapx-odd-lemmas}). Equivalently,
\begin{align}  \label{eq-fact3-pv2}
	\delta^x_P(i,j) +\delta^x_D(i,j) &\geq \delta^x_P(i_{*x},j_{*x}) +\delta^x_D(i_{*x},j_{*x}),	&(i,j)\in T_{*x}
\end{align}
On the other hand, since by (\ref{it-Fx-T}) $F_x$ intersects any Hamiltonian cycle on $V_x$, the tour $T_{*x}$ on $V_x$ intersects $F_x$ 
on some arc $(i,j)$. The optimality of $(i_x,j_x)$ over $F_x$ ensures that $(i,j)$ and $(i_x,j_x)$ satisfy:
\begin{equation}  \label{eq-fact3-pv3}
	\delta^x_P(i_x,j_x) +\delta^x_D(i_x,j_x) \geq \delta^x_P(i,j) +\delta^x_D(i,j)
\end{equation}
We deduce from the two previous inequalities that edge $(i_x,j_x)$ satisfies
\begin{equation}\label{eq-fact3-CS1}
	\delta^x_P(i_x,j_x) +\delta^x_D(i_x,j_x) \geq \delta^x_P(T_*) +\delta^x_D(T_*)
\end{equation}
Together with inequality (\ref{eq-fact3-pv1}), we obtain expression (\ref{eq-fact3}).
\end{proof}

\subsection{Proof of Lemma \ref{lem-Fx}}

In what follows, given $x\in V$, one considers two perfect matchings $M^x_P,M^x_D$ on $V_x$ and the connected components $W_0, W_1,\ldots,W_h$ of the perfect $2$--matching $M^x_P\cup M^x_D$ on $V_x$, where $W_0$ refers to the component that contains the depot vertex provided that $x\neq 0$. Each component $W_s$ induces on $(V_x,M^x_P\cup M^x_D)$ a cycle of even length. 
We describe these cycles by $\{0,v_{0,1},\ldots,v_{0,2m_0+1},0\}$ if $x\neq 0$ and $s=0$, by $\{v_{s,1},\ldots,v_{s,2m_s},v_{s,1}\}$ otherwise.

\subsubsection{Case $x\neq 0$ and $|W_0| =2$}

In this case, $W_0$ induces the cycle $\{0, v_{0, 1}, 0\}$, and $h\geq 1$. 
Since any tour on $V_x$ links vertex 0 to some vertex in $V_x\backslash \{0, v_{0, 1}\}$, we define $F_x$ as
$$\begin{array}{rl}
	F_x	&=\{0\} \times V_x\backslash \{0, v_{0, 1}\}
\end{array}$$
and pick some edge $e_x\in F_x$ that maximizes $\delta^x_P(e)+\delta^x_D(e)$. We assume {\em w.l.o.g.} that $e_x$ is the edge $(0, v_{h, m_h})$. We define $\mathcal{P}_x,N^x_P,N^x_D$ and $\mathcal{P}'_x, N'^x_P, N'^x_D$ as follows:
\begin{itemize}
    \item let $\mathcal{P}_x=\mathcal{P}'_x$ be the loading plan obtained from $\mathcal{P}(M^x_P,M^x_D)$ by loading $x$ at the end of row $2$;
	\item define $N^x_P =N^x_D$ as the edge set obtained from $A(\mathcal{P})$ by substituing for the edge $(v_{h,m_h}, 0)$ the chain $(v_{h,m_h}, x, 0)$;
	\item define $N'^x_P =N'^x_D$ as the edge set obtained from $A'(\mathcal{P})$ by substituing for the edge $(v_{h,m_h+1},0)$ the chain $(v_{h,m_h+1},x,0)$.
\end{itemize}
The fact that $\mathcal{P}_x,N^x_P,N^x_D$, $\mathcal{P}'_x, N'^x_P, N'^x_D$ and $e_x$ satisfy condition (i) of Lemma \ref{lem-Fx} is straightforward from Lemma \ref{lem-dapx-even}. 

\subsubsection{Case $x\neq 0$ and $|W_0|\geq 4$}

In this case, $W_0$ induces on $M^x_P \cup M^x_D$ the cycle $(0,v_{0,1},\ldots,v_{0,2m_0+1},0)$ where $2m_0+1\geq 3$. 
We introduce a new family $\mathcal{Q} =\mathcal{Q}(M^x_P, M^x_D)$ of loading plans given two matchings $M^x_P, M^x_D$ on $V_x$ for some $x\neq 0$, as well as new families $B(\mathcal{Q})$ and $B'(\mathcal{Q})$ of matchings given a loading plan $\mathcal{Q}$. 

$\mathcal{Q}(M^x_P, M^x_D)$ is obtained from $\mathcal{P}(M^x_P,M^x_D)$ by exchanging for each even position $p$ in $\{1 ,\ldots, m_0\}$ the vertex loaded at position $p$ in row 1 with the vertex loaded at position $p$ in row 2. Observe that vertices of $V_x\backslash W_0$ are loaded in $\mathcal{Q}(M^x_P, M^x_D)$ just as the same as in $\mathcal{P}(M^x_P, M^x_D)$.

We describe rows 1 and 2 of $\mathcal{Q}$ by respectively $(i_1,\ldots,i_{\nu +1})$ and $(j_1,\ldots,j_\nu)$. Furthermore, we introduce the cycle $\Gamma =(i_1, i_2 ,\ldots, i_{m_0}, j_{m_0}, j_{m_0 -1} ,\ldots, j_1, i_1)$. We then build two perfect matchings $B(\mathcal{Q})$ and $B'(\mathcal{Q})$ on $V_x$ as follows:
\begin{itemize}
    \item half of the edges of $\Gamma$, including $(i_1,j_1)$, into $B(\mathcal{Q})$, and the other half into $B'(\mathcal{Q})$;
    \item add into $B(\mathcal{Q})$ edges $(i_p, j_p), p\in\{m_0 +1 ,\ldots, (n -3)/2\}$ and $(i_{(n -1)/2}, 0)$;
    \item add into $B'(\mathcal{Q})$ edges $(i_{p +2}, j_p), p\in\{m_0 +1 ,\ldots, (n -5)/2\}$, $(i_{m_0}, i_{m_0 +1})$ and $(j_{(n -3)/2}, 0)$.
\end{itemize}

Using similar arguments as in Lemma \ref{lem-dapx-even}, it is nit too hard to see that the following facts hold:
\begin{fact}\label{lem-dapx-odd}
\begin{enumerate} 
	\item[(i)] $M^x_P \cup B(\mathcal{Q})$ and $M^x_P \cup B'(\mathcal{Q})$ are feasible pickup tours with respect to $\mathcal{Q}$ on $V_x$; 
	\item[(ii)] $M^x_D \cup B(\mathcal{Q})$ and $M^x_D \cup B'(\mathcal{Q})$ are feasible delivery tours with respect to $\mathcal{Q}$ on $V_x$; 
	\item[(iii)] $B(\mathcal{Q})\cup B'(\mathcal{Q})$ is the union of $\Gamma$ and some other cycle $\Gamma'$ over $V_x\backslash W_0 \cup \{0, v_{0, 2m_0 +1}\}$.
\end{enumerate} 
\end{fact}
We omit the proof, but invite the reader to refer to Figure \ref{fig-dapx-Q+B}.

\begin{figure}[t]
\begin{center}
	\includegraphics{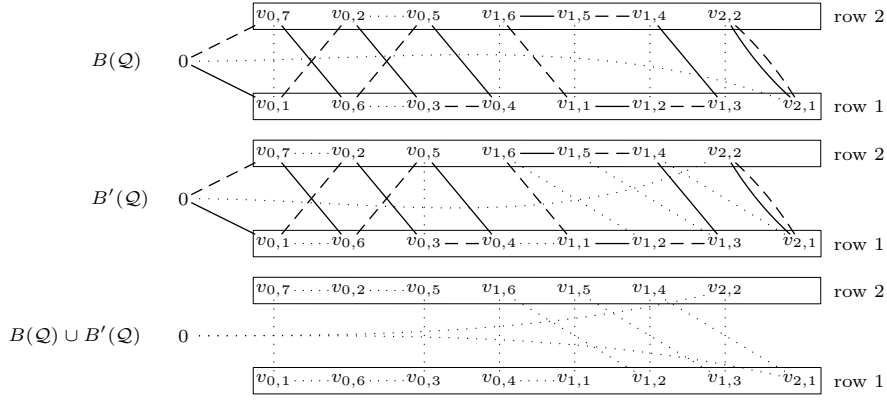}
	\caption{\label{fig-dapx-Q+B}\small{Loading plan $\mathcal{Q}=\mathcal{Q}(M^x_P,M^x_D)$ of $V_x\backslash\{0\}$ and completions $B(\mathcal{Q}),B'(\mathcal{Q})$ (in dotted lines) given two perfect matchings $M^x_P$ (in plain lines) and $M^x_D$ (in dashed lines) on $V_x$}}
\end{center}
\end{figure}

Since any tout on $V_x$ connects a vertex of $V(\Gamma)$ to a vertex of $V_x\backslash V(\Gamma)$, we define $F_x$ as 
$$\begin{array}{rl}
	F_x	&=V(\Gamma) \times \left(V_x\backslash V(\Gamma)\right)
\end{array}$$
and consider an edge $e_x$ that maximizes $\delta^x_P(e)+\delta^x_D(e)$ over $F_x$. By construction, $e_x$ is incident to some vertex $v_{0, j}\in \{v_{0, 1} ,\ldots, v_{0, 2m_0}\}$. We assume {\em w.l.o.g.} $j \leq m_0$. 

We build a first loading plan $\mathcal{P}_x$ and matchings $N^x_P$ and $N^x_D$ as follows:
\begin{itemize}
    \item set $\mathcal{P}_x=\mathcal{Q}(M^x_P,M^x_D)$;
    \item insert $x$ in row 2 at position $p =2$ if $j=1$, at position $p =j$ otherwise;
    \item define $N^x_P =N^x_D$ as the edge set obtained from $B(\mathcal{Q})$ if $j$ is odd, from $B'(\mathcal{Q})$ otherwise, by substituing for the edge $(j_p, j_{p +1})$ the chain $(j_p, x, j_{p +1})$.
\end{itemize}

We build a second loading plan $\mathcal{P}'_x$ and matchings $N'^x_P$ and $N'^x_D$ as follows, depending on $e_x$. Starting with $\mathcal{P}'_x=\mathcal{Q}(M^x_P,M^x_D)$, if $e_x =(v_{0, j}, 0)$, then:
\begin{itemize}
	\item insert $x$ at position $n/2$ in row $2$;
    \item if $p$ is odd, then define $N'^x_P =N'^x_D$ as the edge set obtained from $B'(\mathcal{Q})$ by substituing for the edge $(i_{n/2},0)$ the chain $(i_{n/2},x,0)$; 
	    otherwise, $N'^x_P$ and $N'^x_D$ are obtained from $B(\mathcal{Q})$ by substituing for the edge $(j_{n/2-1},0)$ the chain $(j_{n/2-1},x,0)$.
\end{itemize}
If $e_x=(v_{0,j},v_{0,m_0+1})$, then:
\begin{itemize}
	\item insert $x$ at position $m_0+1$ in row $2$;
    \item if $p$ is odd, then define $N'^x_P =N'^x_D$ as the edge set obtained from $B'(\mathcal{Q})$ by substituing for the edge $(v_{0,m_0+1},x,v_{1,1})$ the chain $(v_{0,m_0+1},x,v_{1,1})$; 
    otherwise, $N'^x_P$ and $N'^x_D$ are obtained from $B(\mathcal{Q})$ by substituing for the edge $(v_{0,m_0+1},x,v_{1,2m_1})$ the chain $(v_{0,m_0+1},x,v_{1,2m_1})$.
\end{itemize}
It remains us to consider the case when $e_x$ is incident to a vertex in $V_x\backslash W_0$. We assume {\em w.l.o.g.} that $e_x =(v_{0,j},v_{h, m_h})$, then:
\begin{itemize}
	\item insert $x$ at position $n/2$ in row $2$;
	\item if $p$ is odd, then define $N'^x_P =N'^x_D$ as the edge set obtained from $B'(\mathcal{Q})$ by substituing for the edge $(i_{n/2},0)$ the chain $(i_{n/2},x,0)$;
	    otherwise, define $N'^x_P =N'^x_D$ as the edge set obtained from $B(\mathcal{Q})$ by substituing for the edge $(j_{n/2-1},0)$ the chain $(j_{n/2-1},x,0)$.
\end{itemize}

The fact that $\mathcal{P}_x,\mathcal{P}'_x,e_x$ and the considered matchings satisfy condition (i) of Lemma \ref{lem-Fx} is straightforward from Fact  \ref{lem-dapx-odd}.

\begin{figure}[t]
\begin{center}
	\includegraphics{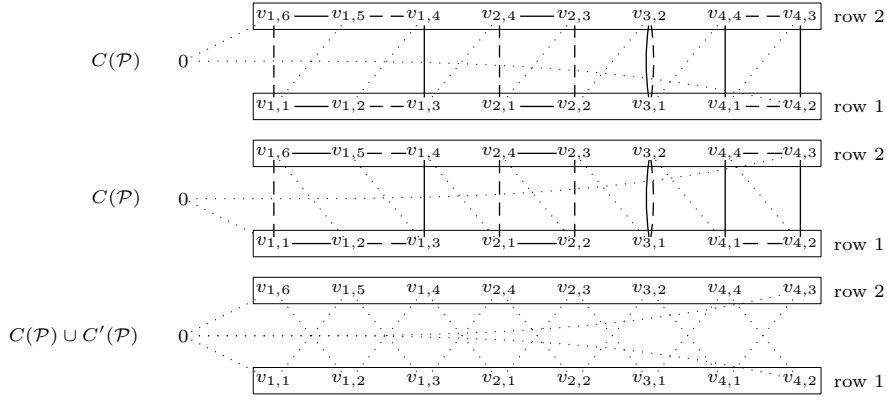}
	\caption{\label{fig-dapx-odd0-P+N}\small{Loading plan $\mathcal{P}=\mathcal{P}(M^0_P,M^0_D)$ of $V\backslash\{0\}$ and completions $C(\mathcal{P}),C'(\mathcal{P})$ (in dotted lines) given two perfect matchings $M^0_P$ (in plain lines) and $M^0_D$ (in dashed lines) on $V_0$}}
\end{center}
\end{figure}

\subsubsection{Case $x=0$ and $h \geq 2$}
Likewise the previous case when $m_0 =0$, we consider for $F_0$ the edge set 
$$\begin{array}{rl}
	F_0	&=\cup_{1\leq s<t\leq h} W_s\times W_t
\end{array}$$
Let $e_0$ be an edge in $F_0$ that maximizes $\delta^0_P(e)+\delta^0_D(e)$. 
We may assume {\em w.l.o.g.} that $e_0$ is the edge $(v_{1,1},v_{h,m_h+1})$ if $|V_0|/2$  is odd, and $(v_{1,1},v_{h,m_h})$ otherwise.

We consider the approximate loading plan $\mathcal{P}_0 =\mathcal{P}(M^0_P,M^0_D)$  Algorithm \ref{algo-apx} returns on $V_0$. Furthermore, similarly to completions $A(\mathcal{P})$ and $A'(\mathcal{P})$, we associate with a loading plan $\mathcal{P}=\left((i_1,\ldots,i_{(n -1)/2}), (j_1,\ldots,j_{(n -1)/2})\right)$ of $V_0$ the two perfect matchings $C(\mathcal{P})$ and $C(\mathcal{P})$ on $V_0$ defined by:
$$\begin{array}{rl}
	C(\mathcal{P})	&=\{(i_p, j_{p +1})\,|\,p =1 ,\ldots, (n -3)/2\} \cup\{(i_{(n-1)/2}, 0)\}\\
	C'(\mathcal{P})	&=\{(j_p, i_{p +1})\,|\,p =1 ,\ldots, (n -3)/2\} \cup\{(j_{(n-1)/2}, 0)\}
\end{array}$$
Observe that the edge sets $C(\mathcal{P}_0) \cup M_P$, $C'(\mathcal{P}_0) \cup M_P$, $C(\mathcal{P}_0) \cup M_D$ and $C'(\mathcal{P}_0) \cup M_D$ all induce on $V$ feasible pickup and delivery tours with respect to $\mathcal{P}$. 
Moreover, if $(n -1)/2$ is odd, then $C(\mathcal{P}_0)\cup C'(\mathcal{P}_0)$ induces on $V_0$ the Hamiltonian cycle
$$\begin{array}{l}
	(i_1,j_2,i_3,j_4,\ldots,i_{(n -1)/2},j_1,i_2,j_3,i_4,\ldots,j_{(n -1)/2},i_1)
\end{array}$$ 
Otherwise, $C(\mathcal{P}_0)\cup C'(\mathcal{P}_0)$ is the union of the two cycles 
$$\begin{array}{ll}
					&(i_1,j_2,i_3,j_4,\ldots,i_{(n -1)/2},j_{(n -1)/2},i_1)\\
	\textrm{and}	&(j_1,i_2,j_3,i_4,\ldots,j_{(n -1)/2},j_{(n -1)/2},j_1).
\end{array}$$
We deduce that $\mathcal{P}_0$, $C(\mathcal{P})$, $C'(\mathcal{P})$ and $e_0$ satisfy (see Figure \ref{fig-dapx-odd0-P+N} for some illustration):
\begin{itemize}
	\item $(\mathcal{P}_0,M^0_P\cup C(\mathcal{P}_0),M^0_D\cup C(\mathcal{P}_0))$ and $(\mathcal{P}_0,M^0_P\cup C'(\mathcal{P}_0),M^0_D\cup C'(\mathcal{P}_0))$ are feasible solutions on $I$;
	\item $C(\mathcal{P}_0)$, $C'(\mathcal{P}_0)$ and $e_0$ satisfy condition (i) of Lemma \ref{lem-Fx}.
\end{itemize}

\begin{figure}[t]
\begin{center}
	\includegraphics{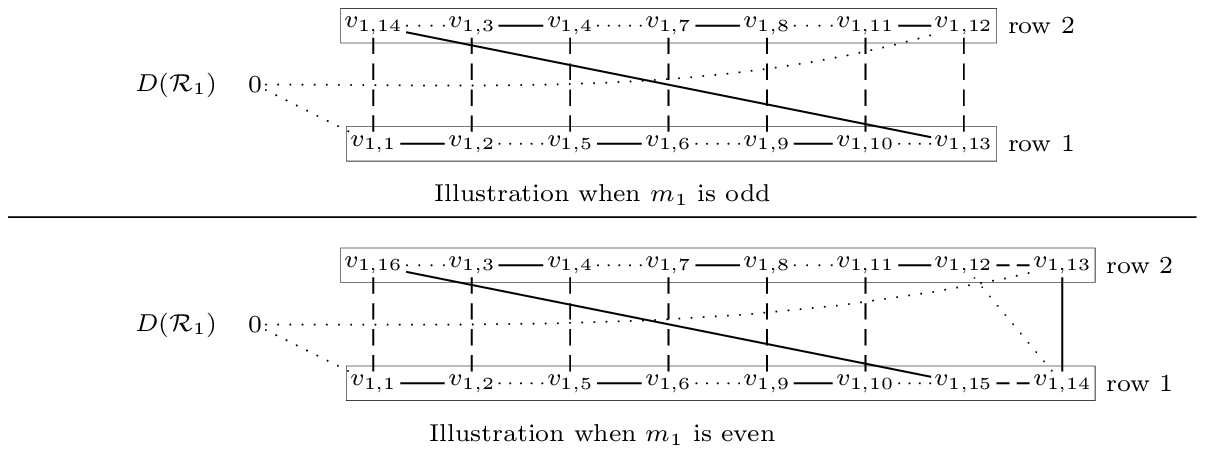}
	\caption{\label{fig-dapx-odd0-P+N-h1}\small{Alternate loading plan $\mathcal{R}_i=\mathcal{R}_i(M^0_P,M^0_D)$ of $V\backslash\{0\}$ an completion $D(\mathcal{R}_i)$ (in dotted lines) given two perfect matchings $M^0_P,M^0_D$ on $V_0$ such that $M^0_P\cup M^0_D$ is a Hamiltonian cycle on $V_0$}}
\end{center}
\end{figure}

\subsubsection{Case $x=0$ and $h =1$}
 %
Let $\{v_1,\ldots,v_m\}$ denote the Hamiltonian cycle $M^0_P\cup M^0_D$ on $V_0$. We consider two families $\mathcal{P}(i),i\in[m]$ and $\mathcal{Q}(i),i\in[m]$ of loading plans on $V_0$, together with their associated completions $N^{\mathcal{P}}_i$ and $N^{\mathcal{Q}}_i$, $i\in[m]$. 
 %
The first family consists of the basic loading plan on $V_0$ and its associated completion $N_2$, but fixing vertex $i$ on coordinates $(1,1)$; namely (indexes are taken $\bmod$ $m$):
\begin{align}
	\mathcal{P}(i)		&=	\left((v_i,\ldots,v_{i+m/2-1}),(v_{i-1},\ldots,v_{i-m/2})\right),					&i\in[m]\label{eq-odd-Pi}\\
	N^{\mathcal{P}}(i)	&=	\cup_{r=1}^{c-2}\left\{\{v_{i-r},v_{i+r}\}\right\}\cup\left\{\{v_{i+m/2},v_i\}\right\},	&i\in[m]\label{eq-odd-NPi}
\end{align}
 %
The loading plans $\mathcal{Q}(i)$ of the second family basically consists, starting with vertex $v_i$, in loading two consecutive vertices of the cycle $\{v_1,\ldots,v_{m},v_1\}$ into alternatively row $1$ and row $2$. Precisely, given $i\in[m]$, let $E_3(i)$ denote the following perfect matching on $V_0$:
\begin{align}
	E_3(i) &=\{\{v_j,v_{j+3}\}\ |\ 1\leq j \leq n-1,\ j\bmod 2\neq i\bmod 2\}\label{eq-E3}
\end{align}
Then, if $m/2$ is odd ({\em iff} $m\bmod 4=2$), the loading plan $\mathcal{Q}(i)$ and its associated completion are defined as (indexes are taken $\bmod$ $m$):
\begin{align}
	\mathcal{Q}(i)		&=
		\left\{\begin{array}{lcl}
			(v_i,\hspace*{0.35cm}	v_{i+1}, v_{i+4}, v_{i+5}	&,\ldots,	&\hspace*{0.8cm}	v_{i-6},v_{i-5},v_{i-2}),\\
			(v_{i-1},				v_{i+2}, v_{i+3},			&,\ldots,	&v_{i-8},			v_{i-7},v_{i-4},v_{i-3})
		\end{array}\right\}\label{eq-odd-Qi-odd}\\
	N^{\mathcal{Q}}(i)	&=	E_3(i)\label{eq-odd-NQi-odd}
\end{align}
Otherwise (thus $m\bmod 4=0$), $\mathcal{Q}(i)$ and $N^{\mathcal{Q}}(i)$ are defined as:
\begin{align}
	\mathcal{Q}(i)		&=
		\left\{\begin{array}{l}
			(v_i,\hspace*{0.35cm}	v_{i+1}, v_{i+4},v_{i+5}			\ ,\ldots, \hspace*{1.06cm}	v_{i-8},v_{i-7},v_{i-2},v_{i-3}),\\
			(v_{i-1},				v_{i+2}, v_{i+3},\hspace*{0.76cm}	,\ldots,\ v_{i-10},			v_{i-9},v_{i-6},v_{i-5},v_{i-4})
		\end{array}\right\}\label{eq-odd-Qi-even}\\
	N^{\mathcal{Q}}(i)	&=E_3(i)	\backslash \left\{ \{v_{i-3},v_i\}, \{v_{i-5},v_{i-2}\}, \{v_{i-7},v_{i-4}\} \right\}\label{eq-odd-NQi-even}\\
							&\hspace*{1.1cm}	\cup \left\{ \{v_{i-5},v_{i-3}\}, \{v_{i-7},v_{i-2}\}, \{v_{i-4},v_i\} \right\}\nonumber
\end{align}
The loading plans $\mathcal{Q}(i)$ and their completion $N^{\mathcal{Q}}(i)$ are depicted in Figures \ref{fig-dapx-odd0-P+N-h1}.

\medskip
We know from the previous analysis that the triple $(\mathcal{P}(i),M^0_P\cup N^{\mathcal{P}}(i),M^0_D\cup N^{\mathcal{P}}(i))$ is a feasible solution for the $\mathsf{DTSPMS_0}$ on $V_0$, $i\in[m]$.
Furthermore, similar arguments enable to establish that $(\mathcal{Q}(i),M^0_P\cup N^{\mathcal{Q}}(i),M^0_D\cup N^{\mathcal{Q}}(i))$ also is a feasible solution of the $\mathsf{DTSPMS_0}$ on $V_0$, $i\in[m]$.
As a consequence, if we define, given $i\in[m]$, $\mathcal{P}_0(i),\mathcal{Q}_0(i),N_0^{\mathcal{P}}(i),N_0^{\mathcal{Q}}(i)$ as
\begin{itemize}
	\item $\mathcal{P}_0(i)=\mathcal{P}(i)$ and $\mathcal{Q}_0(i)=\mathcal{Q}(i)$ if $x=0$, $\mathcal{P}_0(i)$ and $\mathcal{Q}_0(i)$ are obtained from respectively $\mathcal{P}(i)$ and $\mathcal{Q}(i)$ by inserting $x$ at rank $c$ in row $1$ otherwise;
	\item $N^{\mathcal{P}}_0(i)$ is obtained from $N^{\mathcal{P}}(i)$ by inserting $x$ on the edge $\{v_{i+m/2},v_i\}$, $N^{\mathcal{Q}}_0(i)$ is obtained from $N^{\mathcal{Q}}(i)$ by inserting $x$ on the edge $\{v_{i-3},v_i\}$ if $m/2$ is odd, $\{v_{i-4},v_i\}$ otherwise;
\end{itemize}
then we deduce from the feasibility of the solutions that are considered on $V_0$ that $(\mathcal{P}_0(i),M^0_P\cup N^{\mathcal{P}}_0(i),M^0_D\cup N^{\mathcal{P}}_0(i))$ and $(\mathcal{Q}_0(i),M^0_P\cup N^{\mathcal{Q}}_0(i),M^0_D\cup N^{\mathcal{Q}}_0(i))$ are feasible solutions on $I$, $i\in[m]$.
We thus are interested in triples $(\mathcal{P}_0(i),\mathcal{P}_0(j),(i',j'))$, $(\mathcal{P}_0(i),\mathcal{Q}_0(j),(i',j'))$ or $(\mathcal{Q}_0(i),\mathcal{Q}_0(j),(i',j'))$ that satisfy condition (i) of Lemma \ref{lem-Fx}.
We establish:
\begin{claim}\label{claim-odd-NN}
Let $i\neq j$ in $[m]$. In any of the following cases, $N,N'$ and $\{i,j\}$ satisfy condition (i) of Lemma \ref{lem-Fx}:
$$\begin{array}{c|l|cc}
			&\multicolumn{1}{|c|}{m,\ i,j}								&N							&N'\\\hline
	(1) 	&|j-i|\textrm{ is prime with }m								&N_0^{\mathcal{P}}(i)		&N_0^{\mathcal{P}}(j)\\\hline
	(2) 	&m/2\equiv 3\bmod 6\textrm{ and }i\not\equiv j\bmod 3		&N_0^{\mathcal{P}}(i)		&N_0^{\mathcal{Q}}(j)\\\hline
	(3) 	&m/2\bmod 6\in\{1,5\}\textrm{ and }i\not\equiv j\bmod 2		&N_0^{\mathcal{Q}}(i)		&N_0^{\mathcal{Q}}(j)\\
	(4) 	&m/2\equiv 0\bmod 6\textrm{ and }i\not\equiv j\bmod 2\\
	(5) 	&m/2\equiv 4\bmod 6, i\not\equiv j\bmod 2\\
			&\textrm{and }(j-i)\not\equiv\pm 1 \bmod m\\
	(6) 	&m/2\equiv 2\bmod 6\textrm{ and }(j-i)\equiv \pm 1\bmod m\\\hline
	(7) 	&m/2\equiv 2\bmod 6\textrm{ and }i\not\equiv j\bmod 2		&N_0^{\mathcal{Q}}(i+4)		&N_0^{\mathcal{Q}}(j)\\
			&\textrm{and }((j-i)\bmod m)\bmod 6\in\{1,5\}\\\hline
	(8) 	&m/2\equiv 2\bmod 6\textrm{ and }i\not\equiv j\bmod 2		&N_0^{\mathcal{Q}}(i)		&N_0^{\mathcal{Q}}(j+4)\\
			&\textrm{and }((j-i)\bmod m)\bmod 6\in\{3,5\}\\\hline
\end{array}$$
\end{claim}

\begin{proof}
All along the argument, indexes in $[m]$ are taken $\bmod$ $m$.
Preliminary note that $N,N'$ and $\{i,j\}$ satisfy condition (i) of Lemma \ref{lem-Fx} {\em iff}
$\left(N\cup N'\right)\backslash\{\{i,x\},\{x,j\}\}\cup\{\{i,j\}\}$ is a Hamiltonian cycle on $V$ {\em iff}
$\left(N\cup N'\right)\backslash\{\{i,x\},\{x,j\},\{i',x\},\{x,j'\}\}\cup\{\{i,j\},\{i',j'\}\}$ is a Hamiltonian cycle on $V_0$ for the two vertices $i'\neq i$ and $j'\neq j$ such that $\{i',x\}\in N$ and $\{j',x\}\in N'$.

\medskip
{\noindent} $\bullet$ $(1)$: preliminary note that $N^{\mathcal{P}}(1)=N_2$ and $N^{\mathcal{P}}(m)=N_1$.
More generally, assume {\em w.l.o.g.} that $i=m\equiv 0\bmod m$ and let $j$ be some integer in $[m]$. In $N^{\mathcal{P}}(0)$, a vertex $v_h\in V_0$ is adjacent to $v_{-h}$ if $h\notin\{0,m/2\}$, to $v_{h+m/2}$ otherwise. In $N^{\mathcal{P}}(j)$, $v_h=v_{j+(h-j)}$ is adjacent to $v_{j-(h-j)}=v_{2j-h}$ if $h\notin\{j,j+m/2\}$, to $v_{h+m/2}$ otherwise. Hence, starting from $v_0,v_{m/2},v_j$ and $v_{j+m/2}$, $N^{\mathcal{P}}(0)\cup N^{\mathcal{P}}(j)$ generates the sequences:
$$\begin{array}{ll}
	v_0,v_{2j},v_{-2j},v_{4j},v_{-4j},	\ldots	&v_{m/2},v_{2j-m/2},v_{-2j+m/2},v_{4j-m/2},v_{-4j+m/2},\ldots\\
	v_j,v_{-j},v_{3j},v_{-3j},v_{5j},	\ldots	&v_{j+m/2},v_{-j-m/2},v_{3j+m/2},v_{-3j+m/2},v_{5j-m/2},\ldots
\end{array}$$
Since $j$ is prime with $m/2$, $2rj\equiv 0\bmod m$ for some $r\in\mathbb{N}^*$ {\em iff} $rj\equiv 0\bmod(m/2)$ {\em iff} $r$ is a multiple of $m/2$. Now, we may have either $\left((0\equiv m/2) \not\equiv (j\equiv j+m/2)\right) \bmod 2$, or $\left((0\equiv j) \not\equiv (m/2\equiv j+m/2)\right) \bmod 2$, or $\left((0\equiv j+m/2) \not\equiv (m/2\equiv j)\right) \bmod 2$. If the first or the third case occur ({\em iff} $j$ is odd), then $(m/2)j\equiv m/2\bmod m$ whereas if the second case occurs ({\em iff} $j$ is even), then $(m/2)j\equiv 0\bmod m$. $N^{\mathcal{P}}(0)\cup N^{\mathcal{P}}(j)$ therefore takes the following expression depending on the partity of $m/2$ and $j$:
$$\begin{array}{c|c|c}
	N^{\mathcal{P}}(0)\cup N^{\mathcal{P}}(j)	&m/2	&j\\\hline
	\begin{array}{l}
		\{v_0,v_{2j},v_{-2j},v_{4j},\ldots,v_{(m/2)j}=v_{m/2},v_0\}\\
		\cup\ \{v_j,v_{-j},v_{3j},\ldots,v_{(m/2-1)j},v_{-(m/2-1)j}=v_{j+m/2},v_j\}
	\end{array}	&\textrm{even}	&\textrm{odd}\\\hline
	\begin{array}{l}
		\{v_0,v_{2j},v_{-2j},v_{4j},\ldots,v_{(m/2-1)j},v_{-(m/2-1)j}=v_j,\\
			\ v_{j+m/2},v_{-j-m/2},v_{3j+m/2},\ldots,v_{(m/2)j+m/2}=v_{m/2},v_0\}
	\end{array}	&\textrm{odd}	&\textrm{even}\\\hline
	\begin{array}{l}
		\{v_0,v_{2j},v_{-2j},v_{4j},\ldots,v_{(m/2-1)j},v_{-(m/2-1)j}=v_{j+m/2},\\
			\ v_j,v_{-j},v_{3j},\ldots,v_{(m/2)j}=v_{m/2},v_0\}.
	\end{array}	&\textrm{odd}	&\textrm{odd}
\end{array}$$
Hence, if $j$ is odd, then the set $\left(N^{\mathcal{P}}(0)\cup N^{\mathcal{P}}(j)\right)\backslash\{\{v_0,v_{m/2}\},\{v_j,v_{j+m/2}\}\}\cup\{\{v_0,v_j\},\{v_{m/2},v_{j+m/2}\}\}$ is a Hamiltonian cycle on $V_0$.

\medskip
{\noindent} $\bullet$ $(2)$: assume {\em w.l.o.g.} that $i=m$ and thus, $N^{\mathcal{P}}(0)=\{\{v_h,v_{-h}\},h=1,\ldots,m/2-1\}\cup\{\{v_0,v_{m/2}\}\}$. Let $j\in[m]$ such that $j\not\equiv i\bmod 3$. In $E(j)$, $v_h$ is adjacent to $v_{h+3}$ if $h\not\equiv j\bmod 2$ and to $v_{h-3}$ otherwise.
Note that, since $m\bmod 2=m\bmod 3=0$, $h\equiv -h \bmod 2$ and $h\equiv -h \bmod 3$, $h\in[m]$; by contrast, $h\pm 3\not\equiv h \bmod 2, h\in[m]$. $N^{\mathcal{P}}(0)\cup E(j)$ is the union of the two cycles 
$$\begin{array}{l}
	\left\{\begin{array}{ll}
			\{v_0, v_{-3}, v_3, v_6, v_{-6}, \ldots, v_{-m/2}=v_{m/2},v_0\}	&\textrm{if }j\textrm{ is even},\\
			\{v_0, v_3, v_{-3}, v_{-6}, v_6, \ldots, v_{m/2},v_0\}			&\textrm{if }j\textrm{ otherwise}
		\end{array}\right.\\
	\textrm{and }\{v_j, v_{-j}, v_{-j-3}, v_{j+3}, v_{j+6}, \ldots, v_{-j-3(m/3-1)},v_{j+3(m/3-1)}=v_{j-3}, v_j\}.
\end{array}$$
The edge set $\left(N^{\mathcal{P}}(0)\cup N^{\mathcal{Q}}(j)\right)\backslash\{\{v_0,v_{m/2}\},\{v_{j-3},v_j\}\}\cup\{\{v_0,v_j\},\{v_{m/2},v_{j-3}\}\}$ therefore is a Hamiltonian cycle on $V_0$.

\medskip
{\noindent} $\bullet$ $(3)$: The set $E_3$ defined as $E_3=\{\{v_h,v_{h+3}\}\ |\ 1\leq h \leq m\}\equiv (\mathbb{Z}/m\mathbb{Z},\pm 3)$ is a Hamiltonian cycle on $[m]$ {\em iff} $m\not\equiv 0\bmod 3$ {\em iff} $m/2\not\equiv 0\bmod 3$. When $m/2$ is odd and $i\not\equiv j\bmod 2$, then $N^{\mathcal{Q}}(i)\cup N^{\mathcal{Q}}(j)=E_3$ and thus, $N^{\mathcal{Q}}(i)\cup N^{\mathcal{Q}}(j)\backslash\{\{v_{i-3},v_i\},\{v_{j-3},v_j\}\}\cup\{\{v_{i-3},v_{j-3}\},\{v_i,v_j\}\}$ therefore is a Hamiltonian cycle on $V$. 

\medskip
{\noindent} $\bullet$ $(4),(5),(6),(7),(8)$: assume {\em w.l.o.g.} that $i=m\equiv 0\bmod m$. Let thus $j$ be some odd index in $[m]$. First consider the set 
$$\begin{array}{lll}
	E'	&=E_3	&\backslash \left\{ \{v_{-7},v_{-4}\}, \{v_{-5},v_{-2}\}, \{v_{-3},v_0\}	\right\}\\
		&		&\cup		\left\{ \{v_{-7},v_{-2}\}, \{v_{-5},v_{-3}\}, \{v_{-4},v_0\}	\right\}
\end{array}$$
The following Table provides an explicit description of the two edge sets $E_3$ and $E'$ and also identifies three paths $C_0,C_1,C_2$ in $E'$,  depending on $m\bmod 6$:
$$\begin{array}{l}\hline
	\multicolumn{1}{|c|}{m\equiv 0 \bmod 6}\\\hline
	E_3:\ 	\{v_0,v_3,\ldots,v_{-3},v_0\}, \{v_1,v_4,\ldots,v_{-5},v_{-2},v_1\}, \{v_{-1},v_2,\ldots,v_{-7},v_{-4},v_{-1}\}\\
	E':\ 	\{\underbrace{v_0,v_3,\ldots,v_{-3}}_{C_0},
				\underbrace{v_{-5},\ldots,v_4,v_1,v_{-2}}_{C_1},
				\underbrace{v_{-7},\ldots,v_2,v_{-1},v_{-4}}_{C_2},v_0\}\\\hline
	\multicolumn{1}{|c|}{m\equiv 2 \bmod 6}\\\hline
	E_3:\ 	\{v_0,v_3,\ldots,v_{-5},v_{-2},v_1,v_4,\ldots,v_{-7},v_{-4},v_{-1},v_2,v_5,\ldots,v_{-3},v_0\}\\
	E':\ 	\{\underbrace{v_0,v_3,\ldots,v_{-5}}_{C_0},
				\underbrace{v_{-3},\ldots,v_2,v_{-1},v_{-4}}_{C_2},v_0\},
		 \{\underbrace{v_{-2},v_1,v_4,\ldots,v_{-7}}_{C_1},v_{-2}\}\\\hline
	\multicolumn{1}{|c|}{m\equiv 4 \bmod 6}\\\hline
	E_3:\ 	\{v_0,v_3,\ldots,v_{-7},v_{-4},v_{-1},v_2,v_5,\ldots,v_{-5},v_{-2},v_1,v_4,\ldots,v_{-3},v_0\}\\
	E':\ 	\{\underbrace{v_0,v_3,\ldots,v_{-7}}_{C_0},
				\underbrace{v_{-2},v_1,v_4,\ldots,v_{-3}}_{C_1},
				\underbrace{v_{-5},\ldots,v_5,v_2,v_{-1},v_{-4}}_{C_2},v_0\}
\end{array}$$
Now consider $N^{\mathcal{Q}}(0)\cup N^{\mathcal{Q}}(j)$, that is, the edge set 
$$\begin{array}{ll}
	E'	&\backslash	\left\{ \{v_{j-7},v_{j-4}\}, \{v_{j-5},v_{j-2}\}, \{v_{j-3},v_j\}	\right\}\\
		&\cup		\left\{ \{v_{j-7},v_{j-2}\}, \{v_{j-5},v_{j-3}\}, \{v_{j-4},v_j\}	\right\}
\end{array}$$
Although the three edges $\{v_{j-7},v_{j-4}\},\{v_{j-5},v_{j-2}\},\{v_{j-3},v_j\}$ always lie on the paths $C_0,C_1,C_2$, their location and orientation depend on $m\bmod 6$ and $j$; the following Table locates these arcs on $C_0,C_1,C_2$ depending on $m\bmod 6$ and $j$:
$$\begin{array}{l|c|c|c}
															&C_0			&C_1				&C_2\\\hline
	\multicolumn{4}{|c|}{m\equiv 0 \bmod 6}\\\hline
	j\equiv 3\bmod 6 \Leftrightarrow j\in\{3,9,\ldots,-3\}	&(j-3,j)		&(j-2,j-5)			&(j-4,j-7)\\[3pt]
	j\equiv 1\bmod 6 \Leftrightarrow j\in\{1,7,\ldots,-5\}	&(j-7,j-4)		&(j,j-3)			&(j-2,j-5)\\\hline
	\multicolumn{4}{|c|}{m\equiv 2 \bmod 6}\\\hline
	j\equiv 3\bmod 6 \Leftrightarrow j\in\{3,9,\ldots,-5\}	&(j-3,j)		&(j-5,j-2)			&(j-4,j-7)\\
	j\equiv 1\bmod 6,j\neq -1,1								&(j-7,j-4)		&(j-3,j)			&(j-2,j-5)\\
	\ \ \ \ \ \ \ \ \Leftrightarrow	j\in\{7,13,\ldots,-7\}				&				&					&\\\hline
	\multicolumn{4}{|c|}{m\equiv 4 \bmod 6}\\\hline
	j\equiv 3\bmod 6, j\neq -1								&(j-3,j)		&(j-5,j-2)			&(j-4,j-7)\\
	\ \ \ \ \ \ \ \ \Leftrightarrow j\in\{3,9,\ldots,-7\}	&				&					&\\
	j=1														&				&(j-3,j),			&(j-2,j-5)\\
															&				&(j-7,j-4)			&
\end{array}$$
Considering for each of these six cases the way edges $\{v_{j-7},v_{j-2}\}$, $\{v_{j-5},v_{j-3}\}$, $\{v_{j-4},v_j\}$ reconnect the subchains generated by the removal of $\{v_{j-7},v_{j-4}\}$, $\{v_{j-5},v_{j-2}\}$, $\{v_{j-3},v_j\}$ from $E'$, we eventually obverse that $N^{\mathcal{Q}}(0)\cup N^{\mathcal{Q}}(j)$ consists of:
\begin{itemize}
	\item a cycle of the shape $\{v_0,\ldots,v_{j-4},v_j,\ldots,v_{-4},v_0\}$ when $m\equiv 0\bmod 6$ and $j\equiv 3\bmod 6$, or $m\equiv 2\bmod 6$ and $j\neq\pm 1\equiv 1\bmod 6$, the union of two cycles such that the two edges $\{v_0,v_{-4}\}$ and $\{v_j,v_{j-4}\}$ do not belong to the same cycle when $m\equiv 0\bmod 6$ and $j\equiv 1\bmod 6$, or $m\equiv 2\bmod 6$ and $j\equiv 3\bmod 6$, or $m\equiv 4\bmod 6$ and $j=1$; in both the two cases, $N^{\mathcal{Q}}(0)\cup N^{\mathcal{Q}}(j)\backslash\{\{v_{-4},v_0\},\{v_{j-4},v_j\}\}\cup\{\{v_{-4},v_{j-4}\},\{v_0,v_j\}\}$ is a Hamiltonian cycle on $V_0$.
	\item a cycle of the shape $\{v_0,\ldots,v_j,v_{j-4},\ldots,v_{-4},v_0\}$ when $m\equiv 4\bmod 6$ and $j\neq 1\equiv 1\bmod 6$; in this case, $N^{\mathcal{Q}}(0)\cup N^{\mathcal{Q}}(j)\backslash\{\{v_{-4},v_0\},\{v_{j-4},v_j\}\}\cup\{\{v_{-4},v_j\},\{v_0,v_{j-4}\}\}$ is a Hamiltonian cycle on $V$.
\end{itemize}
In order to conclude, finally observe that, given two indexes $i,j\in[m]$, 
	when $m\equiv 0\bmod6$,		$(j-i)\equiv 5\bmod 6$ {\em iff} $(i-j)\equiv 1\bmod 6$; 
	when $m\equiv 2\bmod 6$,	$(j-i)\equiv 5\bmod 6$ {\em iff} $(i-j)\equiv 3\bmod 6$; eventually,
	when $m\equiv 4\bmod 6$, 	$(j-i)\equiv 1\bmod 6$ {\em iff} $(i-j)\equiv 3\bmod 6$ and:
\begin{itemize}
	\item if $(j-i)\equiv 1\bmod 6$ and $j-i\neq 1$, 	then $(j-(i+4))\equiv 3\bmod 6$; otherwise, $(i+4)-j= 3$.
	\item if $(j-i)\equiv 3\bmod 6$ and $j-i\neq -1$, 	then $((j+4)-i)\equiv 1\bmod 6$; otherwise, $(j+4)-i= 3$.
	\item if $(j-i)\equiv 5\bmod 6$, then $(j-(i+4))\equiv 1\bmod 6$ and $((j+4)-i)\equiv 3\bmod 6$.
\end{itemize}
\end{proof}



\comment{ 
    When $h=1$, the set $\{V^q\times V^{q'}\ |\ 1\leq q<q'\leq p_0\}$ is empty. Moreover, we may not deduce a solution of the $\mathsf{DTSPMS}$ on $V$ from a solution on $V\backslash\{0\}$ by inserting vertex $0$ on the edges of $N_P\cup N_D$ that link a vertex loaded at rank $c$ in some row to a vertex loaded at rank $1$ in some row. We thus again reconsider both $F_0$ and the loading strategies.

    Hence, if one inserts $x$ at rank $(n -1)/2$ in row $1$ if $x$ is not the depot vertex, and if one inserts $x$ on the two edges $\{v^p_{m_p/2},v^1_{m_1}\}$ and $\{v^p_{m_p/2+1},v^1_1\}$ of the two perfect matchings $N_1$ and $N_2$, then we obtain a loading plan $\mathcal{P}$ of $V$ and two completions $N_0,N'_0$ such that:
} 